\documentclass[11pt]{article}
\usepackage[margin=1in]{geometry}

\usepackage[utf8]{inputenc}
\usepackage[dvipsnames]{xcolor}
\usepackage[colorlinks, citecolor=RoyalBlue, linkcolor=teal]{hyperref}

\usepackage{booktabs}
\usepackage{tabularx}
\usepackage{array}
\usepackage{mathtools}
\usepackage{makecell}
\usepackage{multirow}
\usepackage{nicefrac}
\usepackage{amsfonts}
\usepackage{graphicx}
\usepackage{subcaption}
\usepackage{amsthm}
\usepackage{amssymb}
\usepackage{amsmath}
\usepackage{pgfplots}
\usepackage{dsfont}
\usepackage[noblocks]{authblk}

\usepackage{tikz}
\usetikzlibrary{
  graphs,
  graphs.standard,
  calc,
  arrows.meta,bending
}

\usepackage{xcolor}
\pgfplotsset{compat=1.17}
\usepackage[ruled,noend]{algorithm2e}
\usepackage{comment}
\usepackage{pifont}
\usepackage{natbib}
\usepackage{wrapfig} 
\LinesNumbered 
\SetAlgoSkip{smallskip}
\usepackage{thm-restate}

\newtheorem{theorem}{Theorem}
\newtheorem{lemma}[theorem]{Lemma}
\newtheorem{claim}[theorem]{Claim}
\newtheorem{corollary}[theorem]{Corollary}
\newtheorem{remark}[theorem]{Remark}

\newtheorem{fact}{Fact}

\newenvironment{claimproof}[1]{\par\noindent\emph{Proof.}\hspace{0.15cm}#1}{\hfill~$\blacktriangleleft$\newline\smallskip}

\newenvironment{mechanism}[1][htb]{%
   \begin{algorithm}[#1]%
  }{\end{algorithm}}
\makeatother

\usepackage{bm}
\renewcommand{\vec}[1]{\boldsymbol{#1}}

\newcommand{\LM}{\textsc{left-median}}
\newcommand{\Leftz}{\textsc{Left-z}}
\newcommand{\RandMed}{\textsc{rand-median}}
\newcommand{\InRange}{\textsc{In-Range}}

\newcommand{\mech}{\ensuremath{\mathcal{M}}}

\newcommand{\set}[1]{\{#1\}}

\newcommand{\median}{\textsc{median}}
\newcommand{\Lmedian}{\textsc{left-median}}

\newcommand{\fl}[1]{\lfloor #1 \rfloor}
\newcommand{\cl}[1]{\lceil #1 \rceil}

\newcommand{\A}{\ensuremath{A}}
\newcommand{\B}{\ensuremath{B}}
\newcommand{\C}{\ensuremath{C}}

\newcommand{\outliers}{\ensuremath{z}}

\newcommand{\opt}{\ensuremath{y^{\star}}}
\newcommand{\alg}{\ensuremath{y}}
\DeclareMathOperator{\OPTSC}{{SC}}
\DeclareMathOperator{\OPTMC}{MC}

\DeclareMathOperator*{\argmin}{arg\,min}

\DeclareMathOperator*{\EX}{\mathbb{E}}

\newcommand\numberthis{\addtocounter{equation}{1}\tag{\theequation}}

\DeclarePairedDelimiter{\floor}{\lfloor}{\rfloor}

\begin{document}

\title{\bfseries Mechanism Design with Outliers and Predictions}

\author[1]{Argyrios Deligkas}
\author[1]{Eduard Eiben}
\author[2,3]{Sophie Klumper}
\author[2,3]{\\ Guido Sch\"afer}
\author[4]{Artem Tsikiridis}

\affil[1]{Royal Holloway, University of London, United Kingdom}
\affil[2]{Centrum Wiskunde \& Informatica (CWI), The Netherlands}
\affil[3]{University of Amsterdam, The Netherlands}
\affil[4]{Technical University of Munich, Germany}

\date{}

\maketitle

\begin{abstract}
\noindent We initiate the study of \emph{mechanism design with outliers}, where the designer can discard $z$ agents from the social cost objective. This setting is particularly relevant when some agents exhibit extreme or atypical preferences. As a natural case study, we consider facility location on the line: $n$ strategic agents report their preferred locations, and a mechanism places a facility to minimize a social cost function. In our setting, the $z$ agents farthest from the chosen facility are excluded from the social cost. While it may seem intuitive that discarding outliers improves efficiency, our results reveal that the opposite can hold.

We derive tight bounds for deterministic strategyproof mechanisms under the two most-studied objectives: utilitarian and egalitarian social cost. Our results offer a comprehensive view of the impact of outliers. We first show that when $z \ge n/2$, no strategyproof mechanism can achieve a bounded approximation for either objective. For egalitarian cost, selecting the $(z + 1)$-th order statistic is strategyproof and 2-approximate. 
In fact, we show that this is best possible by providing a matching lower bound. Notably, this lower bound of 2 persists even when the mechanism has access to a prediction of the optimal location, in stark contrast to the setting without outliers. For utilitarian cost, we show that strategyproof mechanisms cannot effectively exploit outliers, leading to the counterintuitive outcome that approximation guarantees worsen as the number of outliers increases. However, in this case, access to a prediction allows us to design a strategyproof mechanism achieving the best possible trade-off between consistency and robustness. Finally, we also establish lower bounds for randomized mechanisms that are truthful in expectation.
\end{abstract}

\section{Introduction}

You are the coordinator of the annual social event of your department and your task is to choose the venue. Of course, you could decide on your own without asking your colleagues, making you really unpopular within your department.
Ideally, though, you would like to choose a venue that is aligned with the preferences of your colleagues---it is known that everyone wants the venue to be as close as possible to their place.
However, you are facing two major issues you need to overcome.
First, you do not want to be manipulated by your colleagues and you want to incentivize them to declare their {\em true} preferences.
In addition, you know from last year that due to stubborn-Joe---who lives 3 hours away from the department---everyone else had to commute at least one hour to reach the chosen venue. You have decided that this year you will leave out this type of ``{\em outliers}'' in order to make a better decision for the majority of your colleagues. 
After all, not considering some outliers should simplify the problem, shouldn't it?

Problems like the one described above fall into the category of {\em truthful facility location problems}, which have been extensively studied for more than 45 years~\citep{Moulin80}. Furthermore, the seminal paper of~\cite{procaccia2013} established facility location problems as {\em the} paradigm for {\em mechanism design without money}. Since then, a wide variety of models, settings, and mechanisms have been studied; see, for example the survey of~\cite{fl-survey} for an overview of different models. 
In this work, we revisit the two foundational models of truthful facility location, the one of~\cite{Moulin80} and the one of~\cite{procaccia2013}, and study them under the presence of outliers~(see \citep{charikar2001algorithms}). 
Although outliers have been a popular consideration in algorithm design in general, to the best of our knowledge, they have not been considered in the context of mechanism design.
In this paper, we take the first steps towards understanding the impact of outliers on a fundamental mechanism design problem.

\subsection{Our Contribution}

Our main contributions are as follows.

\begin{enumerate}

    \item We introduce the notion of outliers in mechanism design problems. We envision that outliers can be meaningfully incorporated into {\em any} mechanism design problem, which 
    opens up a whole new domain of mechanism design problems with outliers.
    Studying the impact of outliers is not only theoretically appealing but also practically relevant, especially in applications involving agents with extreme or atypical preferences. 
    
    \item We use single facility location on the real line as a first, natural test case for our setting of mechanism design with outliers. We derive tight bounds for deterministic strategyproof mechanisms for the two most-studied objectives, i.e., utilitarian and egalitarian social cost. We provide a complete picture of the impact of outliers and our results reveal some counter-intuitive phenomena. We also extend our analysis to randomized mechanisms. 
    
    \item We further enrich our model by incorporating output predictions (e.g., obtained through machine-learning techniques), contributing to the recent emerging line of research on learning-augmented mechanism design. 
    We derive a mechanism with an optimal consistency-robustness trade-off for the utilitarian objective and, unlike the problem without outliers, we prove an impossibility result for the egalitarian objective. 

\end{enumerate}
Next, we elaborate on these three points in more detail.

\begin{figure}[t]
  \centering
  \begin{tikzpicture}[scale=0.65]
    \begin{axis}[
      width=\textwidth,
      height=0.5\textwidth,
      xlabel={$z$}, ylabel={$f(n,z)$},
      xmin=1, xmax=9, ymin=1, ymax=10,
      ytick={1,2,3,4,5,6,7,8,9},
      xtick={1,2,3,4,5,6,7,8,9},
      grid=major,
      clip=false,
      legend style={at={(rel axis cs:0.03,0.95)}, anchor=north west, draw, rounded corners, fill=white, font=\small},
      legend columns=1,
    ]

      \addplot+[red!80!black, very thick, mark=*, mark options={fill=white},
                samples at={1,2,...,9}, domain=1:9]
        {20/(20-2*x)};
      \addlegendentry{\(n=20\)}

      \addplot+[cyan!80!black, very thick, mark=*, mark options={fill=white},
                samples at={1,2,...,7}, domain=1:7]
        {(16)/(16-2*x)};
      \addlegendentry{\(n=16\)}

      \addplot+[blue!80!black, very thick, mark=*, mark options={fill=white},
                samples at={1,2,...,5}, domain=1:5]
        {(12)/(12-2*x)};
      \addlegendentry{\(n=12\)}

      \addplot+[green!50!black, very thick, mark=*, mark options={fill=white},
                samples at={1,2,...,3}, domain=1:3]
        {8/(8-2*x)};
      \addlegendentry{\(n=8\)}

    \end{axis}
  \end{tikzpicture}
  \caption{Illustration of the best possible approximation guarantee $f(n,z) = \frac{n}{n-2z}$ for the utilitarian objective for $n$ (even) agents and $z$ outliers. For $z > \floor{\frac{n-1}{2}}$, the problem is inapproximable.}
  \label{fig:graph-approx}
\end{figure}
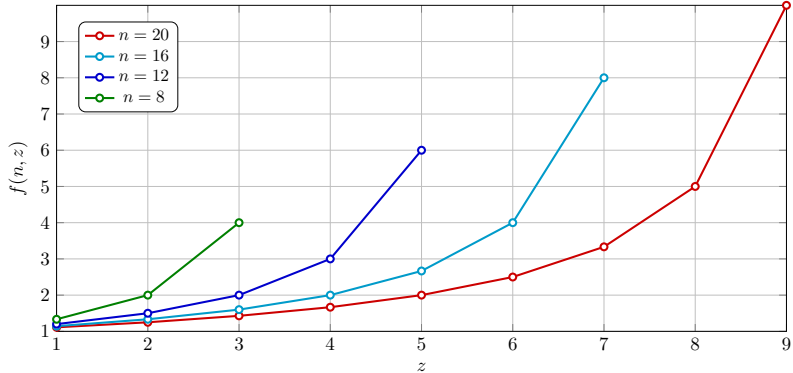

\paragraph{New Domain: Mechanism Design with Outliers.} 

In a mechanism design problem with outliers, the objective is altered in that it does not aim to optimize over the set of \emph{all} agents, but rather it only accounts for a certain number of them.
More specifically, the input of the problem includes an extra integer parameter $z$ that denotes the number of agents that are excluded from the objective function. 
Put differently, if there are $n$ agents in total, the objective only accounts for the ``best'' set of $n-z$ agents, i.e., for minimization objectives those that contribute the least to it and for maximization objectives those that contribute the most to it; formal definitions will be given below. 
Thus, this perspective of disregarding $z$ outliers can be universally applied to any objective function and domain. 

\paragraph{Facility Location with Outliers.} 

In facility location terms, the problem with outliers reduces to the following setting. We have $n$ agents, where each one of them has a private location on the real line. Our goal is to choose a single location on the line for a facility such that a social cost objective, depending on the distances between the facility and the locations of the agents, is minimized, while ensuring that the agents declare their true locations. 
We consider the two most prominently studied social cost objectives, namely the \emph{utilitarian social cost} (i.e., total cost of the agents) and the \emph{egalitarian social cost} (i.e., maximum cost of the agents). However, in our setting with outliers, the contributions of the $z$ agents with the largest costs are disregarded from the social cost objective; it is not hard to see that these agents will be the ones farthest away from the facility.

Although it may seem intuitively plausible that outliers should only improve the efficiency of a mechanism---since we are free to disregard some of the agents---our results reveal that the opposite is true.
For example, the achievable approximation guarantee for utilitarian social cost gets worse as the number of outliers increases (see Figure \ref{fig:graph-approx}).
The crux of the issue is that a strategyproof mechanism \emph{cannot} simply ignore the $z$ outliers chosen by the optimal solution.
A partial overview of the bounds we derive in this paper is given in Table~\ref{tab:overview}. 

Our first set of technical results fully characterizes the landscape of deterministic strategyproof mechanisms with outliers for both the utilitarian and egalitarian objective.
We begin by showing that if we can discard $z \ge \nicefrac{n}{2}$ of the agents, then there is no deterministic strategyproof mechanism that achieves a bounded approximation. This impossibility result holds for both objectives.  
Then, we turn our attention to the egalitarian objective and prove that the mechanism choosing the \emph{$(z+1)$-th order statistic}, i.e., the $(z+1)$-th smallest location reported by the agents, as the location of the facility is strategyproof and 2-approximate. We show that this is tight for every $z < \nicefrac{n}{2}$ by providing a matching lower bound.
After that, we focus on the utilitarian objective. 
To this end, we consider the strategyproof mechanism which always chooses the location of the median agent.
Analyzing the approximation is technically more involved, and we derive an upper bound as a function of the number of agents $n$ and outliers $z$ (see Table~\ref{tab:overview}). 
Finally, we also show that this is tight by providing a matching lower bound. 

We also extend our analysis to randomized mechanisms to understand how much randomization can help to improve our bounds. For the utilitarian objective, we show that a natural randomized extension of our deterministic strategyproof mechanism leads to an improved approximation guarantee when the number of agents is even, and we establish additional lower bounds. In contrast, for the egalitarian objective, we show that randomization does not seem to offer an advantage, by proving a lower bound that matches the approximation of our deterministic mechanism for three agents.

\begin{table}[t]
    \centering
    \setlength{\tabcolsep}{4pt}
    {
    \begin{tabularx}{\textwidth}{@{}>{\raggedright\arraybackslash}X@{} *{2}{>{\centering\arraybackslash}X} *{2}{>{\centering\arraybackslash}X}@{}}
        \toprule
         \textbf{\footnotesize Objective}& \multicolumn{2}{c}{\textbf{\footnotesize Deterministic Mechanisms}} & \multicolumn{2}{c}{\textbf{\footnotesize Randomized Mechanisms}} \\
        \cmidrule(lr){2-3} \cmidrule(l){4-5}
        & \textbf{\footnotesize Upper Bound} & \textbf{\footnotesize Lower Bound} & \textbf{\footnotesize Upper Bound} & \textbf{\footnotesize Lower Bound} \\
        \midrule
        \textbf{\footnotesize Utilitarian}
            & \makecell[c]{$\frac{n}{n-2z}\mid\frac{n-1}{n-2z+1}$\\\footnotesize(Thm.~\ref{th:SCdeter})}
            & \makecell[c]{$\frac{n}{n-2z}\mid\frac{n-1}{n-2z+1}$\\\footnotesize(Thm.~\ref{lem:SC:LB-deterministic})}
            & \makecell[c>{\centering}X]{$\frac{n^2 - 2nz + 2z}{(n - 2z)(n - 2z +2)}$ $\mid$ -- \\\footnotesize(Thm.~\ref{lem:SC:rand:UB})} 
            & \makecell[c]{$\tfrac{3}{2}\mid 2$\\\footnotesize(Thm. \ref{thm:SC:LB-3over2-randomized} \& \ref{thm:SC:LB-2-randomized-odd})} \\
        \addlinespace
        \textbf{\footnotesize Egalitarian}
            & \makecell[c]{$2$\\\footnotesize(Thm.~\ref{thm:MC:sp-2approx})}
            & \makecell[c]{$2$\\\footnotesize(Thm.~\ref{lem:MC:LB-2-deterministic})}
            & \makecell[c]{$2$\\\footnotesize(due to Thm.~\ref{thm:MC:sp-2approx})}
            & \makecell[c]{$2$\\\footnotesize(Thm.~\ref{lem:MC:LB-2-deterministic})} \\
        \bottomrule
    \end{tabularx}
    }
    \caption{Partial overview of our bounds. Note that some bounds depend on the parity of $n$; we use ($n$ even $|$ $n$ odd) to distinguish between these cases. Our randomized lower bounds hold for specific values of $n$ and $z$.}
    \label{tab:overview}
\end{table}

\paragraph{Integrating Predictions.} 

We extend our model by incorporating predictions. In this setting, we assume that the mechanism has access to some (possibly erroneous) prediction of the optimal location of the facility (i.e., with respect to the social cost objective adapted to the outlier setting). 
This kind of prediction, termed \emph{output advice} in \citep{christodoulou2024mechanism}, has recently become one of the standard benchmarks in mechanism design with predictions. 
We note that from an information-theoretic perspective, this augmentation is minimal.
In our setting, the predicted optimal facility location represents an aggregate of the private locations of $n-z$ agents, and can be learned from historical or customer data.\footnote{However, how this prediction is learned from data in practice is beyond the scope of this paper.}

Our second set of technical results examines how mechanisms can leverage these predictions in the setting of facility location with outliers to achieve improved approximation guarantees. 
More precisely, we are interested in designing mechanisms that achieve the best possible trade-off between \emph{consistency} and \emph{robustness}. While consistency demands good approximation guarantees when the prediction is accurate, robustness ensures that the approximation guarantee does not deteriorate arbitrarily when the prediction is erroneous. 
Surprisingly, we derive strong impossibility results for the egalitarian objective: we prove that no strategyproof mechanism can achieve a bounded robustness and a consistency better than $2$ (Theorem \ref{lem:MC:LB-prediction-opt}). In light of our $2$-approximate strategyproof mechanism with outliers mentioned above, this basically means that predictions do not help at all in this case. 
However, for the utilitarian objective and for $z \le \nicefrac{n}{3}$ outliers, we obtain positive results. We derive a strategyproof mechanism, called \InRange, that is $1$-consistent and achieves the best possible robustness guarantee (Theorem \ref{th:SC-UB-n>=3z}).
Additionally, we show that the approximation of \InRange\ smoothly interpolates between these two extremes depending on a natural error parameter quantifying the accuracy of the prediction.
If the number of outliers is $\nicefrac{n}{3} < z < \nicefrac{n}{2}$, we show that no strategyproof mechanism can achieve 1-consistency and a bounded robustness (Theorem \ref{lem:SC:LB-prediction-z>1-unbounded}).

\subsection{Related Work}

\paragraph{Facility location.} 
A big variety of different models have been proposed that studied
the number of facilities whose location needs to be determined~\citep{procaccia2013,Lu2010two-facility,fotakis2014two}, 
agents with different types of preferences for the facilities (optional~\citep{chen2020max,kanellopoulos2023discrete,li2020constant,serafino2016}, fractional~\citep{fong2018fractional}, or hybrid~\citep{feigenbaum2015hybrid})
and obnoxious facilities~\citep{cheng2013obnoxious}.
Furthermore, there exist models with other limitations or features: the facilities can only be built at specific fixed locations~\citep{feldman2016voting,gai2024mixed,kanellopoulos2025truthful,Xu2021minimum}; there are limited resources that allow only some of the available facilities to be built~\citep{deligkas2023limited}; there is limited available information during the decision process~\citep{chan2023ordinal,filos2024distributed}.
One specific model that 
looks similar to ours is the {\em capacitated} facility location problem \citep{AzizCLP20,aziz2020facility,auricchio2024capacitated,auricchio2024facility}. 
In this model, there is one facility (or more), and each facility is associated with a capacity, i.e., an upper bound on the number of agents it can serve. There, part of the mechanism design problem is to {\em choose} which agents will be served by a facility. In our case, though, although we have a ``capacity'' on the number of agents we consider in the optimal solution, we do not exclude any agent from being served by the facility.

\paragraph{Mechanism Design with Predictions.} A prominent ``beyond-worst-case analysis'' framework (see, e.g.,~\citep{roughgarden21}) is the design and analysis of algorithms in environments augmented with predictions. Often referred to as the \emph{learning-augmented framework}, this approach aims to overcome worst-case lower bounds by leveraging predictions provided, e.g., by a machine learning algorithm trained on historical data. Ever since the work of \citet{mahdian12} on online advertising, this framework has remained popular in the area of online algorithms and beyond. We refer the reader to the repository of \citet{algorithms_with_predictions} for a comprehensive list of related works. More recently, \citet{xu21} and \citet{agrawal2024} simultaneously proposed the study of strategic environments augmented with predictions about private information or aggregations thereof. In fact, strategic facility location can be viewed as the paradigmatic problem for this setting as it is studied in both papers, with \citet{agrawal2024} deriving deterministic, strategyproof mechanisms that achieve the best-possible consistency-robustness trade-offs in two dimensions for both the utilitarian and egalitarian objectives. 
There have been multiple follow-up works on strategic facility location in this environment \citep{barak2024mac,balkanski2024randomized,chen24}, including the work of \citet{christodoulou2024mechanism} who also propose a universal error notion which we adopt in this work.
The framework has also been applied to other mechanism design settings without money, but also to settings with monetary transfers. For example, \citet{colini2024trust} study assignment problems, fairness is considered in \citep{cohen24}, and \citet{ratsikas25} consider social choice in environments with predictions. For auction-related environments, see \citep{balkanski24online, balkanski23scheduling, prasad23,lu24,caragiannis24}.

\paragraph{Clustering with Outliers.}
The problem of optimization with outliers was introduced by \cite{charikar2001algorithms}, and they studied several well-known clustering problems ---$k$-median, $k$-means, facility location--- in high dimensions with outliers and provided polynomial-time algorithms that achieve constant approximations.  
Building upon this, there was a long line that studied the abovementioned problems and improved the approximation bounds using a variety of techniques ranging from local search to iterative rounding and parameterized algorithms~\citep{almanza2022,Chen_outliers,friggstad2019approximation,dabas2021variants,krishnaswamy_k-means_outliers,zhang2021local,agrawal2023clustering,maity2024linear,feng2019improved,goyal_et_al_2020,gupta_local_search}.

\section{Preliminaries} \label{sec:prelims}

We start by introducing the non-strategic problem of single facility location on the line, which follows the model of~\cite{charikar2001algorithms}.

\paragraph{Facility Location with Outliers.}
In the standard single facility location problem on the real line, we must determine the location $y \in \mathbb{R}$ of a facility that serves a set $N = \{ 1, 2, \ldots, n\}$ of $n$ agents. 
Each agent $i \in N$ has a location $x_i \in \mathbb{R}$ and if $y$ is chosen as the location of the facility, agent $i$ incurs a cost that is equal to the distance between $i$'s location and the location of the facility, i.e., $|y - x_i|$.
Two well-studied objectives are minimizing the \emph{utilitarian} social cost, which is the total cost of all agents, i.e., $\sum_{i \in N} |y - x_i|$, and minimizing the \emph{egalitarian} social cost, which is the maximum cost of an agent, i.e., $\max_{i \in N} |y - x_i|$. 
However, scenarios exist for which it is desirable to only consider a subset of the agents in the objective, and we refer to this problem as the problem with \emph{outliers}. 
In this case, for any $n \ge 2$, we are additionally given an integer parameter\footnote{Note that if $\outliers = n$, both objectives are equal to 0, and if $\outliers = 0$, we retrieve the problem without outliers.} $\outliers \in \{ 1, 2, \ldots, n-1\}$ representing the number of outliers. 
The goal is still to find a location $y$ that minimizes the social cost objective, but only for a set of $n-\outliers$ agents which can be chosen freely. 
Let the profile $\vec{x} = (x_i)_{i \in N} \in \mathbb{R}^n$ denote the vector of locations of the $n$ agents. 
If the location $y$ is chosen, then the set of outliers can be computed by solving
the following optimization problems for the utilitarian and egalitarian objective,
respectively

\begin{equation} \label{eq:OPT}
    \OPTSC(y, \vec{x},\outliers) := \min_{\substack{S \subset N: \\|S| = n-z}} \ \sum_{i \in S} |y - x_i|
\quad \quad \text{ and } \quad \quad 
    \OPTMC(y, \vec{x},\outliers) := \min_{\substack{S \subset N:\\|S| = n-z}} \ \max_{i \in S} |y - x_i|.
\end{equation}
We will use $\OPTSC^{\star}(\vec{x},z) = \min_{y \in \mathbb{R}}\OPTSC(y, \vec{x},z)$ and $\OPTMC^{\star}(\vec{x},z) = \min_{y \in \mathbb{R}}\OPTMC(y, \vec{x},z)$ to denote the optimal value of the objectives.
Moreover, we use $S^{\star}(\vec{x},z)$ to denote an optimal set of non-outliers and $y^{\star}(\vec{x}, z)$ to denote an optimal location, i.e., for the utilitarian objective $y^{\star}(\vec{x},z) \in \argmin_{y \in \mathbb{R}}\OPTSC(y, \vec{x},z)$.\footnote{Note that an optimal solution is not necessarily unique; if this is the case, we will clarify which optimal solution we refer to.} 
Given a profile $\vec{x}$, we will use $\sigma$ to refer to the indices in increasing order of value, i.e., $x_{\sigma(1)} \le x_{\sigma(2)} \le \ldots \le x_{\sigma(n)}$. 

\paragraph{Strategic Agents.}
In the strategic version of the problem, the preferred location $p_i \in \mathbb{R}$ of an agent $i \in N$ is {\em private information} and only known to agent $i$. 
Therefore, each agent $i \in N$ declares a preferred location $x_i \in \mathbb{R}$. Given a profile of declared locations $\vec{x} \in \mathbb{R}^n$ and the number of outliers $\outliers \in \{ 1, 2, \ldots, n-1 \}$, we seek a \textit{mechanism} $\mathcal{M}$ that chooses the location $y \in \mathbb{R}$ of the facility, i.e., $\mathcal{M}: (\mathbb{R}^{n},\mathbb{N}^+) \rightarrow \mathbb{R}$. 
If $y$ is chosen as the location of the facility, agent $i$ incurs a cost of $|y- p_i|$. 
As the agents are strategic, they will misreport their preferred location if this reduces their incurred cost.
We therefore seek \textit{strategyproof} mechanisms in which no agent can reduce their incurred cost by misreporting their preferred location.  
More formally, a mechanism $\mathcal{M}$ is strategyproof if for each agent $i \in N$, any $p_i, x_i \in \mathbb{R}$ and any $\vec{x}_{-i} \in \mathbb{R}^{n-1}$, it holds that $|\mathcal{M}((p_i, \vec{x}_{-i}),z) - p_i| \le |\mathcal{M}((x_i, \vec{x}_{-i}),z) - p_i|$. 
Note that we do not require a mechanism to determine which agents are the outliers and all $n$ agents can still make use of the facility.
Namely, given a profile $\vec{x}$, the number of outliers $z$ and the location $y$ chosen by the mechanism, the set of non-outliers minimizing the objective can easily be computed.  

\cite{Moulin80} established the following characterization result for deterministic strategyproof mechanisms which will be useful in our setting as well. 

\begin{theorem}[\cite{Moulin80}] \label{thm:characterization}
A deterministic mechanism $\mech$ for the facility location problem on the line is strategyproof if and only if there exist $n+1$ real numbers $\alpha_1, \dots, \alpha_{n+1} \in \mathbb{R} \cup \set{-\infty, +\infty}$, called \emph{phantom points}, such that 
\begin{equation}\label{eq:med-char}
\forall \vec{x} = (x_1, \dots, x_n) \in \mathbb{R}^n: \qquad \mech(\vec{x}) = \median(x_1, \dots, x_n, \alpha_1, \dots, \alpha_{n+1}).
\end{equation}
\end{theorem}

Note that the above class of mechanisms in particular also contains the class of \emph{$k$-order statistic mechanisms} for any $k \in [n]$: by defining $n+1-k$ many phantom points to be $-\infty$ and the others to be $+\infty$, $\median(x_1, \dots, x_n, \alpha_1, \dots, \alpha_{n+1})$ always corresponds to the $k$-th smallest element of $\vec{x}$. 
We will use this observation several times in our proofs. 

\begin{remark}
A few remarks are in order about the characterization result of \cite{Moulin80}.
\begin{enumerate}

    \item \citet{Moulin80} focuses exclusively on mechanisms that are \emph{anonymous}, i.e., the identity of the agents does not affect the outcome. We make the same assumption throughout this paper, but do not mention it explicitly.
    Further, the characterization result above also holds for mechanisms satisfying the stronger incentive compatibility notion of \emph{group-strategyproofness}.
    
    \item The original characterization result by \citet{Moulin80} is proven for agents with single-peaked preferences. Our setting of Euclidean distances on the real line is a special case of this. However, in terms of characterization result, it is known that Theorem~\ref{thm:characterization} remains valid for this case as well 
    (see \citep{Border1983} and the discussion in \citep{Peters1993}). 

    \item \citet{Moulin80} also shows that the class of all mechanisms that are strategyproof and Pareto efficient is described by an analogous characterization with the only difference being that one needs to consider $n-1$ phantom points $\alpha_1, \dots, \alpha_{n-1}$ only. This result is not relevant in our setting as we care about social cost objectives (rather than Pareto efficiency). 

\end{enumerate}
\end{remark}

The following corollary also follows from Theorem \ref{thm:characterization} and will be useful when proving lower bounds for both objectives. 

\begin{corollary} \label{lem:SP-property-M}
Consider a profile $\vec{x}$ and $z$ outliers and let $\mathcal{M}$ be a strategyproof mechanism with $\mathcal{M}(\vec{x},\outliers) = \alg$.
Then, for any $i \in N$ with $x_i < \alg$ or $\alg < x_i$, it holds that $\mathcal{M}((x'_i, \vec{x}_{-i}),\outliers) = \alg$ for all $x'_i \in [x_i, \alg]$ and $x'_i \in [\alg,x_i] $, respectively. 
\end{corollary}

Given the utilitarian objective, we say that a mechanism $\mathcal{M}$ is $\rho$-approximate, with $\rho \ge 1$, if for any input $(\vec{x},z)$ it holds that $\OPTSC(\mathcal{M}(\vec{x}, z), \vec{x}, z) \le \rho \cdot \OPTSC^{\star}(\vec{x},z)$. 
The definition for the egalitarian objective is defined analogously. 
For notational convenience, we use $\OPTSC(\mathcal{M}(\vec{x}, z))$ and $\OPTMC(\mathcal{M}(\vec{x}, z))$ instead of $\OPTSC(\mathcal{M}(\vec{x}, z), \vec{x}, z)$ and $\OPTMC(\mathcal{M}(\vec{x}, z), \vec{x}, z)$, respectively, in the remainder of the paper.

In this work, we also consider randomized mechanisms, which allow us to randomly determine the location of the facility based on the input. Formally, given some input $(\vec{x}, z)$, a randomized mechanism $\mech(\vec{x},z)$ is a probability distribution over $\mathbb{R}$, and we write $y \sim \mech(\vec{x},z)$ to denote its (random) location. Furthermore, both the objectives in \eqref{eq:OPT}, as well as the costs that agents incur, extend naturally to randomized mechanisms by being evaluated in expectation, i.e., $\EX_{y \sim \mech(\vec{x},z)}[\OPTSC(\vec{x}, z)]$, $\EX_{y \sim \mech(\vec{x},z)}[\OPTMC(\vec{x}, z)]$, and $\EX_{y \sim \mech(\vec{x},z)}[|y-p_i|]$ for each agent $i \in N$. There are two variants of strategyproofness considered in the literature: \emph{strategyproofness-in-expectation} and \emph{universal strategyproofness}. The first notion requires that every agent $i \in N$ satisfies $
\EX_{y \sim \mech((p_i,\vec{x}_{-i}),z)}[|y-p_i|] \le \EX_{y \sim \mech((x_i,\vec{x}_{-i}),z)}[|y-p_i|],
$ for every $x_i \in \mathbb{R}$ and every $\vec{x}_{-i} \in \mathbb{R}^{n-1}$. On the other hand, the second notion is more stringent and requires that the randomized mechanism is a probability distribution over strategyproof deterministic mechanisms.

\paragraph{Impossibility Result.} Note that if $\outliers = n-1$, both objectives considered in this paper minimize the cost of a single agent, which can trivially be solved optimally by any mechanism deterministically choosing a $k$-th order statistic. 
However, if the number of outliers is at least half the number of agents, we cannot hope to achieve any bounded approximation, as we show in the lemma below.\footnote{Note that the lemma does not apply to $n \le 3$, i.e., for $n = 3$, $\cl{\frac{3}{2}} = 2$ already exceeds $n-2 = 1$.}

\begin{restatable}{theorem}{twoone} \label{lem:approx-unbounded-z>=halfn} 
Consider the utilitarian or egalitarian social cost objective and let $n \ge 4$ and $\outliers \in \{ \cl{\frac{n}{2}}, \cl{\frac{n}{2}} + 1, \ldots, n-2 \}$. Then, there is no deterministic strategyproof mechanism that achieves a bounded approximation guarantee.
\end{restatable}

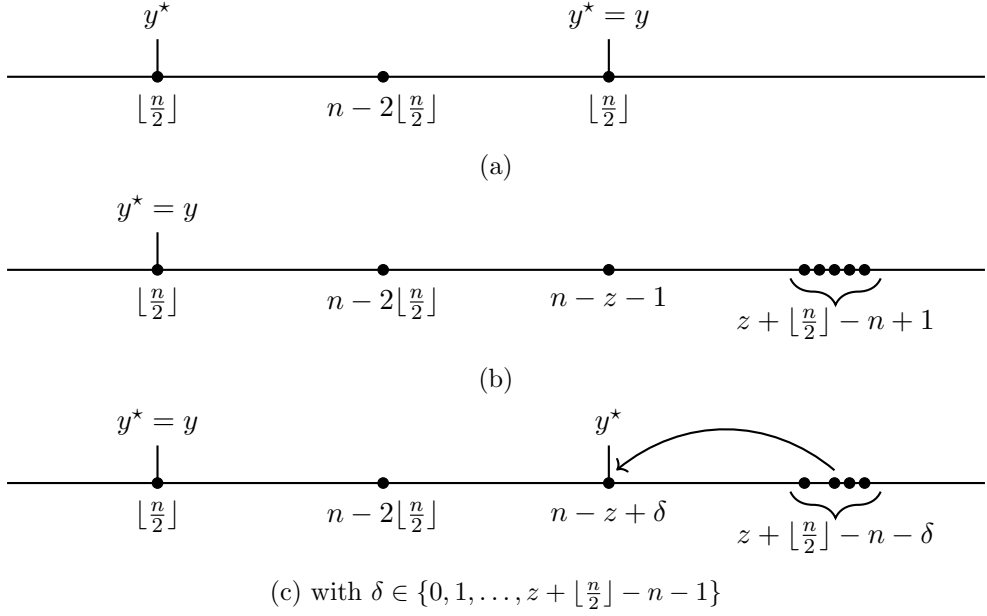
\begin{figure}[t]
\centering
\begin{subfigure}[b]{0.9\linewidth}
    \centering
    \begin{tikzpicture}
        \draw[thick] (0,0) -- (13,0);
        \draw[thick] (2,0) -- (2,0.5) node[above]{$y^{\star}$};
        \filldraw[black] (2,0) circle (2pt) node[below, yshift=-0.1cm]{$\fl{\frac{n}{2}}$};
        \filldraw[black] (5,0) circle (2pt) node[below, yshift=-0.1cm]{$n - 2\fl{\frac{n}{2}}$}; 
        \draw[thick] (8,0) -- (8,0.5) node[above]{$y^{\star} = \alg$};
        \filldraw[black] (8,0) circle (2pt) node[below, yshift=-0.1cm]{$\fl{\frac{n}{2}}$}; 
    \end{tikzpicture}
    \caption{ }
    \label{fig:ApproxUnboundedA}
\end{subfigure}

\begin{subfigure}[b]{0.9\linewidth}
    \centering
    \begin{tikzpicture}
        \draw[thick] (0,0) -- (13,0);
        \draw[thick] (2,0) -- (2,0.5) node[above]{$y^{\star} = \alg$};
        \filldraw[black] (2,0) circle (2pt) node[below, yshift=-0.1cm]{$\fl{\frac{n}{2}}$};
        \filldraw[black] (5,0) circle (2pt) node[below, yshift=-0.1cm]{$n - 2\fl{\frac{n}{2}}$}; 
        \filldraw[black] (8,0) circle (2pt) node[below, yshift=-0.1cm]{$n- \outliers -1$}; 
        \filldraw[black] (10.6,0) circle (2pt);
        \filldraw[black] (10.8,0) circle (2pt);
        \filldraw[black] (11,0) circle (2pt); 
        \filldraw[black] (11.2,0) circle (2pt);
        \filldraw[black] (11.4,0) circle (2pt);
        \draw [thick, decorate,decoration={brace,amplitude=10pt,mirror},xshift=0.4pt,yshift=-0.4pt](10.4,-0.1) -- (11.6,-0.1) node[black,midway,yshift=-0.6cm] {$\outliers + \fl{\frac{n}{2}} - n + 1$};
    \end{tikzpicture}
    \caption{ }
    \label{fig:ApproxUnboundedB}
\end{subfigure}

\begin{subfigure}[b]{0.9\linewidth}
    \centering
    \begin{tikzpicture}
        \draw[thick] (0,0) -- (13,0);
        \draw[thick] (2,0) -- (2,0.5) node[above]{$y^{\star} = \alg$};
        \filldraw[black] (2,0) circle (2pt) node[below, yshift=-0.1cm]{$\fl{\frac{n}{2}}$};
        \filldraw[black] (5,0) circle (2pt) node[below, yshift=-0.1cm]{$n - 2\fl{\frac{n}{2}}$}; 
        \filldraw[black] (8,0) circle (2pt) node[below, yshift=-0.1cm]{$n-\outliers + \delta$};
        \draw[thick] (8,0) -- (8,0.5) node[above]{$y^{\star}$};
        \filldraw[black] (10.6,0) circle (2pt);
        \filldraw[black] (11,0) circle (2pt);
        \filldraw[black] (11.2,0) circle (2pt);
        \filldraw[black] (11.4,0) circle (2pt);
        \draw [thick, decorate,decoration={brace,amplitude=10pt,mirror},xshift=0.4pt,yshift=-0.4pt](10.4,-0.1) -- (11.6,-0.1) node[black,midway,yshift=-0.6cm] {$\outliers + \fl{\frac{n}{2}} - n - \delta$};
       \draw[->, thick, bend right=40] (11,0.17) to (8.1,0.17);
    \end{tikzpicture}
    \caption{ with $\delta \in \{0, 1, \ldots, \outliers + \fl{\frac{n}{2}} - n -1\}$}
    \label{fig:ApproxUnboundedC}
\end{subfigure}
\caption{Profiles used in the proof of Theorem \ref{lem:approx-unbounded-z>=halfn}. The numbers below the locations (dots) indicate the number of agents with this location.}
\label{fig:ApproxUnbounded}
\end{figure}

\begin{proof}
Let $n \ge 4$ and $\outliers \ge \cl{\frac{n}{2}}$ and towards a contradiction, assume that there exists a strategyproof mechanism $\mathcal{M}$ that achieves a bounded approximation guarantee.
The core idea behind our negative result is that for a profile with an optimal solution of $0$ social cost, any mechanism that achieves a bounded approximation must also choose an optimal location.
We will leverage this in our proof by considering a profile with a unique optimal location of $0$ social cost, and sequentially construct profiles with each time only changing the location of one agent.
To this end, consider the profiles depicted in Figure \ref{fig:ApproxUnbounded}, in which the numbers below the locations (dots) indicate the number of agents with that location.

First, consider the profile in Figure \ref{fig:ApproxUnboundedA}. 
In this case, there are two possible optimal locations $y^{\star}$ with a social cost of $0$, namely the leftmost and rightmost location. 
To see this, note that there is at most $n-2\fl{\frac{n}{2}} \le 1$ agent located in the middle, so this location together with the leftmost (or the rightmost) cluster can be disregarded for the social cost as $\outliers \ge \cl{\frac{n}{2}}$.
Therefore, $\mathcal{M}$ will have to output one of these two optimal locations in order to achieve a bounded approximation guarantee. Assume w.l.o.g. that $\mathcal{M}$ outputs the rightmost location as depicted in Figure \ref{fig:ApproxUnboundedA}.

Secondly, consider the profile in Figure \ref{fig:ApproxUnboundedB}. In this case, there is only one optimal location $y^{\star}$ with a social cost of $0$, namely the leftmost location. 
To see this, note that this is the only location equal to at least $n-z$ locations of the agents. 
Therefore, $\mathcal{M}$ will have to output this optimal location in order to achieve a bounded approximation guarantee, as depicted in Figure \ref{fig:ApproxUnboundedB}.

Now consider the profile in Figure \ref{fig:ApproxUnboundedC} with $\delta=0$. 
One agent with a location in the rightmost cluster in Figure \ref{fig:ApproxUnboundedB} now has a location in the second rightmost cluster. 
Note that in this case, the second rightmost cluster is also an optimal location as it contains $n-z$ locations.
However, as $\mathcal{M}$ is strategyproof, it must still place the facility at the leftmost cluster by Corollary \ref{lem:SP-property-M}.
The same argument holds when considering $\delta = 1, 2, \ldots, \outliers + \fl{\frac{n}{2}} - n -1$ consecutively.

Finally, consider Figure \ref{fig:ApproxUnboundedC} with only one agent with a location in the rightmost cluster, i.e., $\delta = z + \fl{\frac{n}{2}} - n -1$. 
Then, if agent $i$ would declare that their location is in the cluster to the left of their location, this would lead to the profile depicted in Figure \ref{fig:ApproxUnboundedA} for which $\mathcal{M}$ places the facility in the rightmost cluster.  
This contradicts that $\mathcal{M}$ is strategyproof by Corollary \ref{lem:SP-property-M}, as the location of the facility should remain unchanged after this deviation of agent $i$. 
\end{proof}

In light of this impossibility result and optimality for $z = n-1$, we only consider profiles with $n \ge 3$ agents and $z$ outliers such that $1 \le \outliers \le \fl{\frac{n-1}{2}}$ in the remainder of this paper. 

\paragraph{Augmenting Mechanisms with Predictions.}

In the setting with predictions, besides $\vec{x}$ and $z$, a mechanism is additionally given a predicted location $\hat{y}$ of the optimal facility as input.
We use $(\vec{x}, z, \hat{y})$ to refer to input augmented with a prediction.
Given input $(\vec{x}, z, \hat{y})$, we say that $\hat{y}$ is a \emph{perfect prediction} if $\hat{y}$ corresponds to a location of an optimal facility, i.e., for the utilitarian objective this holds if $\OPTSC(\hat{y}, \vec{x}, z)= \OPTSC^{\star}(\vec{x},z)$.
Depending on the quality of the prediction, we consider the following approximation notions for the utilitarian objective\footnote{The definitions for the egalitarian objective are defined analogously.}, which are standard in the predictions literature (see e.g., \citep{lykouris21}). 
\begin{itemize}
    \item \textit{Consistency:} A mechanism $\mathcal{M}$ is \emph{$\alpha$-consistent}, with $\alpha \ge 1$, if for any input $(\vec{x}, z, \hat{y})$ with a perfect prediction, it holds that $\OPTSC(\mathcal{M}(\vec{x}, z, \hat{y})) \le \alpha \cdot \OPTSC^{\star}(\vec{x},z)$.
    
    \item \textit{Robustness:} A mechanism $\mathcal{M}$ is \emph{$\beta$-robust}, with $\beta \ge 1$, if for any input $(\vec{x}, z, \hat{y})$ with an arbitrary prediction, it holds that $\OPTSC(\mathcal{M}(\vec{x}, z, \hat{y})) \le \beta \cdot \OPTSC^{\star}(\vec{x},z)$.
\end{itemize}
\noindent
To define an approximation notion beyond the two extremes of a perfect prediction and any arbitrary prediction, we define an error that measures the quality of the prediction (see, e.g., \citep{christodoulou2024mechanism,colini2024trust,gkatzelis2025clock}), as follows:\footnote{If $\OPTSC^{\star}(\vec{x},\outliers) =0$, we define $\eta(\vec{x},\outliers,\hat{y}) = 1$ if $\OPTSC(\hat{y}, \vec{x}, z) = 0$ and $ \eta(\vec{x},\outliers,\hat{y}) = \infty$ otherwise.}
\begin{equation} \label{eq:error}
    \eta(\vec{x},\outliers,\hat{y}) = \frac{\OPTSC(\hat{y}, \vec{x}, z)}{\OPTSC^{\star}(\vec{x},\outliers)}.
\end{equation}
Given this definition, input $(\vec{x}, z, \hat{y})$ with a perfect prediction has a prediction error of 1.
As the quality of the prediction deteriorates, the error measure increases, possibly to $\infty$. 
Our goal is to construct a mechanism that achieves an approximation guarantee that smoothly interpolates between the consistency and robustness guarantees as a function of the error measure.
Formally, a mechanism $\mathcal{M}$ is \emph{$f(\eta)$-approximate}, with $f(\eta) \ge 1$, if for any input $(\vec{x}, z, \hat{y})$ with a prediction error of at most $\eta$, i.e., $\eta(\vec{x},\outliers,\hat{y}) \le \eta$, it holds that $\OPTSC(\mathcal{M}(\vec{x}, z, \hat{y})) \le f(\eta) \cdot \OPTSC^{\star}( \vec{x},z)$.

\section{Minimizing the Egalitarian Objective with Outliers} \label{sec:max-cost}

In this section, we focus on the egalitarian objective and prove a tight bound of 2 for deterministic strategyproof mechanisms.

An optimal solution to the non-strategic problem for a profile $\vec{x}$ and $\outliers$ outliers disregards the locations $x_i$ in the objective that are among the smallest and largest values. 
This could be $\outliers$ of the smallest values, i.e., all locations $x_i$ with $i=\sigma(k)$ and $k < \outliers +1$, or $\outliers$ of the largest values, i.e., all locations $x_i$ with $i = \sigma(k)$ and $k > n-\outliers$, or any combination in between. 
Therefore, in order for a mechanism $\mathcal{M}(\vec{x},\outliers) = \alg$ to achieve a bounded approximation guarantee, $\mech$ must choose a location $\alg$ such that $x_{\sigma(\outliers+1)} \le y \le x_{\sigma(n- \outliers)}$. 
For example, if $\mech$ would output $\alg < x_{\sigma(\outliers+1)}$, it could be that $x_{\sigma(\outliers)} < x_{\sigma(\outliers+1)} = x_{\sigma(\outliers+2)} = \ldots = x_{\sigma(n)} = y^{\star}$, leading to an optimal egalitarian social cost of 0 by disregarding the $\outliers$ leftmost locations.
But as $\outliers \le \fl{\frac{n-1}{2}}$, the location $\alg$ chosen by $\mech$ has a positive egalitarian social cost as there are at least $z +1$ agents with a location greater than $y$, i.e., $x_{\sigma(k)} - \alg > 0$ for $k \in \{ z+1, z+2, \ldots, n\}$, leading to an unbounded approximation guarantee.

Therefore, consider the mechanism \Leftz\ that chooses the $(\outliers+1)$-th order statistic as the location of the facility\footnote{In fact, any choice of the $k$-th order statistic with $\outliers + 1 \le k \le n-\outliers$ would work.}: 
\begin{equation*}
\Leftz(\vec{x}, \outliers) = x_{\sigma(\outliers+1)}.
\end{equation*}

\begin{restatable}{theorem}{threeone} \label{thm:MC:sp-2approx} 
Let $\outliers \le \fl{\frac{n-1}{2}}$. Then, mechanism \Leftz\ is strategyproof and has an approximation guarantee of $2$ for the egalitarian objective. 
\end{restatable}

\begin{proof}
Strategyproofness of \Leftz\ follows from Theorem \ref{thm:characterization}.
For the approximation guarantee, fix an arbitrary profile $(\vec{x}, z)$ and let $y=\Leftz(\vec{x},z)$. 
For the remainder of the proof, we fix an optimal solution and use $S^{\star}:=S^{\star}(\vec{x},z)$ as the optimal set of non-outliers with respect to $y^{\star}:=y^{\star}(\vec{x}, z)$. The following claim will be useful for our analysis.
\begin{claim}\label{claim:mc-main-claim}
    $\max_{i \in S^*}|y-x_i| \leq 2 \cdot \OPTMC^{\star}(\vec{x}, z)$. 
\end{claim}
\begin{claimproof}
We distinguish two cases, depending on the relative values of $y$ and $y^{\star}$.

\smallskip
\noindent \underline{Case 1:} $y^{\star} \le \alg$. Note that $\OPTMC^{\star}(\vec{x}, \outliers) \ge \alg - y^{\star}$, as there are at least $\outliers$ locations to the right of $\alg = x_{\sigma(\outliers +1)}$, i.e., $\alg - y^{\star} \le x_{i} - y^{\star}$ for $i \in \{ \sigma(\outliers+2), \sigma(\outliers +3), \ldots, \sigma(n)\}$.
For an illustrative example, see Figure \ref{fig:maxCost2approx}. 
For locations $x_{i} \le y^{\star}$ with $i \in S^{\star}$, it holds that $y - x_i = y^{\star} - x_i + y - y^{\star} \le 2 \cdot \OPTMC^{\star}(\vec{x}, \outliers)$.
For locations $x_{i} \ge \alg$ with $i \in S^{\star}$, it holds that $x_i  - \alg \le x_i - y^{\star} \le \OPTMC^{\star}(\vec{x}, \outliers)$.
Finally, for locations $y^{\star} < x_{i} < \alg$ with $i \in S^{\star}$, it holds that $\alg - x_i \le \alg - y^{\star} \le \OPTMC^{\star}(\vec{x}, \outliers)$.

\smallskip
\noindent \underline{Case 2:} $y^{\star} > \alg$. This case is symmetrical. Again note that $\OPTMC^{\star}(\vec{x}, \outliers) \ge y^{\star} - \alg$, as there are at least $\outliers$ locations to the left of $\alg = x_{\sigma(\outliers +1)}$, i.e., $y^{\star} - \alg \le y^{\star} - x_{i}$ for $i \in \{ \sigma(1), \sigma(2), \ldots, \sigma(\outliers) \}$.
For locations $x_{i} \ge y^{\star}$ with $i \in S^{\star}$, it holds that $x_i - \alg = x_i - y^{\star} + y^{\star} - \alg \le 2 \cdot \OPTMC^{\star}(\vec{x}, \outliers)$.
For locations $x_{i} \le \alg$ with $i \in S^{\star}$, it holds that $\alg - x_i \le y^{\star} - x_i \le \OPTMC^{\star}(\vec{x}, \outliers)$.
Finally, for locations $\alg < x_{i} < y^{\star}$ with $i \in S^{\star}$, it holds that $x_i - \alg \le y^{\star} - \alg \le \OPTMC^{\star}(\vec{x}, \outliers)$.
\end{claimproof}

\noindent
We complete the proof by observing that
\begin{equation*}
    \OPTMC(y ,\vec{x}, z) = \min_{\substack{S \subset N:\\|S|=n-z}}\max_{i \in S}|y-x_i| \leq \max_{i \in S^*}|y-x_i| \leq 2 \cdot \OPTMC^{\star}(\vec{x},z).
\end{equation*}
Here, the first inequality follows by the fact that $|S^{\star}|=n-z$, i.e., $S^{\star}$ is a feasible solution to \eqref{eq:OPT} for $(y, \vec{x},z)$, and the second inequality is due to Claim \ref{claim:mc-main-claim}. This concludes the proof.
\end{proof}

\begin{figure}[h]
\centering
    \begin{tikzpicture}
        \draw[thick] (2,0) -- (12,0); 
        \filldraw[black] (3,0) circle (2pt) node[below, yshift=-0.1cm]{$x_j$} ; 
         \draw[stealth-stealth,thick](3,0.25) -- (9,0.25)  node[color=black, above, xshift=-2.8cm]{$\alg - x_{j} = y^{\star} - x_{j} + \alg - y^{\star}$};
        \draw[stealth-stealth,thick](3,0.25) -- (7,0.25) ; \draw[stealth-stealth,thick](7,0.25) -- (9,0.25) ;
        \draw[thick] (7,0) -- (7,-0.5) node[below]{$y^{\star}$};
        \filldraw[black] (7.5,0) circle (2pt) node[below, yshift=-0.1cm]{$x_k$} ; 
        \filldraw[black] (9,0) circle (2pt);
        \draw[thick] (9,0) -- (9,-0.5) node[below]{\alg};
        \draw[stealth-stealth,thick](7,1) -- (11,1)  node[color=black, above, xshift=-0.75cm]{$x_{\ell} - y^{\star}$};
         \draw[stealth-stealth,thick](9,0.25) -- (11,0.25)  node[color=black, above, xshift=-0.75cm]{$x_{\ell} - \alg$};
        \filldraw[black] (11,0) circle (2pt) node[below, yshift=-0.1cm]{$x_{\ell}$} ; 
    \end{tikzpicture}
\caption{Example of Case 1 of Theorem \ref{thm:MC:sp-2approx}. } 
\label{fig:maxCost2approx}
\end{figure}
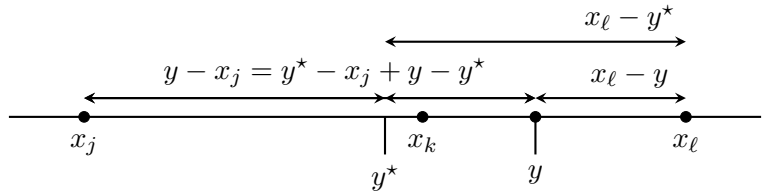

We conclude this section by proving that, for the egalitarian objective, \Leftz\ achieves the best-possible approximation guarantee among all deterministic strategyproof mechanisms. 
Moreover, there are indications that this result may hold even in a stronger sense: as we show in Theorem \ref{lem:MC:LB-2-deterministic}, not even a mechanism that is strategyproof in expectation can achieve an expected approximation guarantee better than 2 for $n = 3$ and $z = 1$.

\begin{restatable}{theorem}{threetwoone}  \label{lem:MC:LB-2-deterministic}
Let $\outliers \le \fl{\frac{n-1}{2}}$. Then, there is no deterministic strategyproof mechanism with an approximation guarantee better than $2$ for the egalitarian objective.  Also, for $n = 3$ and $z = 1$, there is no randomized mechanism which is strategyproof in expectation and achieves an expected approximation guarantee better than $2$ for the egalitarian objective. 
\end{restatable}

We prove each of the two statements in Theorem \ref{lem:MC:LB-2-deterministic} separately in Lemma \ref{lemma:lb-mc-det} and Lemma \ref{lemma:lb-mc-rand}.

\begin{lemma}\label{lemma:lb-mc-det}
Let $\outliers \le \fl{\frac{n-1}{2}}$. Then, there is no deterministic strategyproof mechanism with an approximation guarantee better than $2$ for the egalitarian objective. 
\end{lemma}

\begin{proof}
The proof is inspired by the lower bound of \cite{procaccia2013}.
Let $\outliers \le \fl{\frac{n-1}{2}}$ and $\varepsilon>0$. 
Towards a contradiction, assume that there exists a strategyproof mechanism $\mathcal{M}$ that is $(2 - \varepsilon)$-approximate.
We again consider a sequence of profiles, which are depicted in Figure \ref{fig:MC-LB-2}. 

The profile in Figure \ref{fig:MC-LB-2-a} is such that there are $z$ agents with a location of 0 (leftmost cluster), $n- 2z$ agents with a location of $\frac{1}{2}$, and $\outliers$ agents with a location of 1 (rightmost cluster).
Note that $n - 2 \outliers \ge 1$ as $\outliers \le \fl{\frac{n-1}{2}}$. 
Also note that all agents from either the leftmost or rightmost cluster can be disregarded in the objective. 
Therefore, there are two optimal locations $y^{\star}$, namely at $\frac{1}{4}$ and $\frac{3}{4}$, with an egalitarian social cost of $\frac{1}{4}$. 
As $\mathcal{M}$ is $(2 - \varepsilon)$-approximate, it must be that $\alg$ is located in $[\frac{\varepsilon}{4}, \frac{1}{2} - \frac{\varepsilon}{4}]$ or in $[\frac{1}{2} + \frac{\varepsilon}{4}, 1 - \frac{\varepsilon}{4}]$, depicted by the curly brackets in Figure \ref{fig:MC-LB-2-a}. 
Assume w.l.o.g. that $\mathcal{M}$ chooses a location $\alg$ in $[\frac{\varepsilon}{4}, \frac{1}{2} - \frac{\varepsilon}{4}]$. 

Now, consider the profile in Figure \ref{fig:MC-LB-2-b} with $\delta =1$: the location of one agent from the right cluster moved from 1 to $\frac{1}{2}$. 
Note that as $\mathcal{M}$ is strategyproof, $\alg$ remains unchanged by Corollary \ref{lem:SP-property-M}. 
The same reasoning holds when considering $\delta = 2, 3, \ldots, \outliers -1$ consecutively. 

Finally, consider the profile in Figure \ref{fig:MC-LB-2-c} in which all the locations of the $\outliers$ agents from the right cluster moved to $\frac{1}{2}$. 
Note that $\alg$ remains unchanged as $\mathcal{M}$ is strategyproof (Corollary \ref{lem:SP-property-M}). 
However, as the left cluster now contains the locations of $\outliers$ agents, all agents from this cluster can be disregarded in the objective. 
Therefore, the only optimal location $y^{\star}$ is the location of the cluster at $\frac{1}{2}$ with an egalitarian social cost of 0, as depicted in Figure \ref{fig:MC-LB-2-c}. 
This contradicts that $\mathcal{M}$ is $(2 - \varepsilon)$-approximate as $\alg \in [\frac{\varepsilon}{4}, \frac{1}{2} - \frac{\varepsilon}{4}]$, so $\alg$ has a positive egalitarian social cost, concluding the proof. 
\end{proof}

\begin{figure}[ht]
\centering
\begin{subfigure}[b]{0.37\linewidth}
    \centering
    \begin{tikzpicture}
        \draw[thick] (-0.5,0) -- (4.5,0);
        \filldraw[black] (0,0) circle (2pt) node[below, yshift=-0.1cm]{ $\outliers$};
        \filldraw[black] (4,0) circle (2pt) node[below, yshift=-0.1cm]{ $\outliers$}; 
        \filldraw[black] (2,0) circle (2pt) node[below, yshift=-0.1cm]{ $n - 2 \outliers$}; 
        \draw[thick] (1,0) -- (1,-0.5) node[below, xshift=0.1cm]{$y^{\star}$};
        \draw[thick] (3,0) -- (3,-0.5) node[below, xshift=0.1cm]{$y^{\star}$};
        \draw [thick, decorate,decoration={brace,amplitude=10pt},xshift=0.4pt,yshift=0.8pt](0.2,0.1) -- (1.8,0.1) node[black,midway,yshift=0.6cm] {$\alg$};
        \draw [thick, decorate,decoration={brace,amplitude=10pt},xshift=0.4pt,yshift=0.8pt](2.2,0.1) -- (3.8,0.1) node[black,midway,yshift=0.6cm] {$\alg$};
    \end{tikzpicture}
    \caption{ }
    \label{fig:MC-LB-2-a}
\end{subfigure}
\begin{subfigure}[b]{0.37\linewidth}
    \centering
    \begin{tikzpicture}
        \draw[thick] (-0.5,0) -- (4.5,0);
        \filldraw[black] (0,0) circle (2pt) node[below, yshift=-0.1cm]{$\outliers$};
        \filldraw[black] (4,0) circle (2pt) node[below, yshift=-0.1cm]{$\outliers - \delta$}; 
        \filldraw[black] (2,0) circle (2pt) node[below, yshift=-0.1cm]{$n - 2 \outliers + \delta$}; 
        \draw[thick] (1,0) -- (1,-0.5) node[below, xshift=0.1cm]{$y^{\star}$};
        \draw[thick] (3,0) -- (3,-0.5) node[below, xshift=0.1cm]{$y^{\star}$};
        \draw [thick, decorate,decoration={brace,amplitude=10pt},xshift=0.4pt,yshift=0.8pt](0.2,0.1) -- (1.8,0.1) node[black,midway,yshift=0.6cm] {$\alg$};
         \draw[->, thick, bend right=40] (3.9,0.17) to (2.1,0.17);
    \end{tikzpicture}
    \caption{$\delta \in \{ 1, 2, \ldots, \outliers -1\}$ }
    \label{fig:MC-LB-2-b}
\end{subfigure}
\begin{subfigure}[b]{0.24\linewidth}
    \centering
    \begin{tikzpicture}
        \draw[thick] (-0.5,0) -- (2.5,0);
        \filldraw[black] (0,0) circle (2pt) node[below, yshift=-0.1cm]{$\outliers$};
        \filldraw[black] (2,0) circle (2pt) node[below, yshift=-0.1cm]{$n - \outliers$}; 
        \draw[thick] (2,0) -- (2,0.5) node[above, xshift=0.1cm]{$y^{\star}$};
        \draw [thick, decorate,decoration={brace,amplitude=10pt},xshift=0.4pt,yshift=0.8pt](0.2,0.1) -- (1.8,0.1) node[black,midway,yshift=0.6cm] {$\alg$};
    \end{tikzpicture}
    \vspace{0.5cm}
    \caption{ }
    \label{fig:MC-LB-2-c}
\end{subfigure}
\caption{Profiles used in the proof of Lemma \ref{lemma:lb-mc-det}. }
\label{fig:MC-LB-2}
\end{figure}
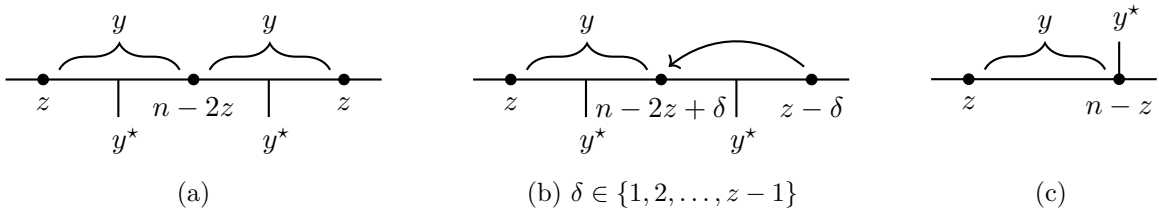

\begin{lemma}\label{lemma:lb-mc-rand}
    For $n = 3$ and $z = 1$, there is no randomized mechanism which is strategyproof in expectation and achieves an expected approximation guarantee better than $2$ for the egalitarian objective.  
\end{lemma}

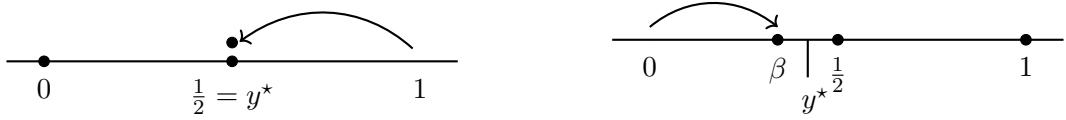
\begin{figure}
\centering
\begin{subfigure}[b]{0.48\linewidth}
    \centering
    \begin{tikzpicture}
        \draw[thick] (-1,0) -- (5,0);
        \filldraw[black] (-0.5,0) circle (2pt) node[below, yshift=-0.1cm]{$0$};
        \filldraw[black] (4.5,0) circle (0.00002pt) node[below, yshift=-0.1cm]{$1$}; 
        \filldraw[black] (2,0) circle (2pt) node[below, yshift=-0.1cm]{$\frac{1}{2}=y^{\star}$}; 
        \filldraw[black] (2,0.25) circle (2pt) node[below, yshift=0.2cm]{$ $}; 

         \draw[->, thick, bend right=40] (4.4,0.17) to (2.1,0.26);
    \end{tikzpicture}
    \caption{Deviation under $\vec{p}$ for agent $r$ to $\nicefrac{1}{2}$}
    \label{fig:MC-LB-RAND-b1}
\end{subfigure}
\begin{subfigure}[b]{0.48\linewidth}
    \centering
    \begin{tikzpicture}
        \draw[thick] (-1,0) -- (5,0);
        \filldraw[black] (-0.5,0) circle (0.002pt) node[below, yshift=-0.1cm]{$0$};
        \filldraw[black] (4.5,0) circle (2pt) node[below, yshift=-0.1cm]{$1$}; 
        \filldraw[black] (2,0) circle (2pt) node[below, yshift=-0.1cm]{$\frac{1}{2}$};
        \filldraw[black] (1.2,0) circle (2pt) node[below, yshift=-0.1cm]{$\beta$}; 

        \draw[thick] (1.6,0) -- (1.6,-0.5) node[below, xshift=0.1cm]{$y^{\star}$};
         \draw[->, thick, bend left=40] (-0.5,0.17) to (1.2,0.17);
    \end{tikzpicture}
    \caption{Deviation under $\vec{p}$ for agent $\ell$ to $\beta$ ``close'' to $\nicefrac{1}{2}$.}
    \label{fig:MC-LB-RAND-b2}
\end{subfigure}
\caption{Profiles used in the proof of Lemma \ref{lemma:lb-mc-rand}.}
\label{fig:MC-LB-RAND}
\end{figure}

\begin{proof}
Towards a contradiction, suppose that $\mathcal{M}$ is a randomized mechanism that is strategyproof in expectation and $(2-\varepsilon)$-approximate for $\varepsilon \in (0,1]$. Consider the profile with $N = \{\ell, m, r\}$. Suppose the preferred locations are $p_{\ell} = 0$, $p_{m} = \frac{1}{2}$ and $p_r = 1$.

First, suppose that the rightmost agent $r$ misreports their location as $\frac{1}{2}$, while the other two agents report their preferred locations truthfully (see Figure~\ref{fig:MC-LB-RAND-b1}). 
Observe that $\OPTMC^{\star}((\frac{1}{2}, \vec{p}_{-r}),1) = 0$, as under the profile $(\frac{1}{2}, \vec{p}_{-r})$, agent $\ell$ is the outlier. Thus, the optimal location is to place the facility at $\frac{1}{2}$, the declared location of all the other agents. 
Moreover, it is crucial to observe that $\mech$ must follow suit, i.e., it must be that $\mech((\frac{1}{2}, \vec{p}_{-r}),1)=\frac{1}{2}$ since, by the assumed approximation guarantee of $\mech$, it must be that $\EX[\mech((\frac{1}{2}, \vec{p}_{-r}),1)]\leq (2-\varepsilon) \cdot 0$. Thus, we obtain
\begin{equation}\label{eq:deviation-of-r-lb-mb}
        \frac{1}{2}= \EX_{y \sim \mech((\frac{1}{2}, \vec{p}_{-r}), 1)}\left[|y- 1| \right]= \EX_{y \sim \mech((\frac{1}{2}, \vec{p}_{-r}), 1)}\left[|y- p_{r}| \right] \geq \EX_{y \sim \mech(\vec{p}, 1)}\left[|y- p_{r}| \right], 
\end{equation}
with the inequality following from the fact that $\mech$ is strategyproof in expectation.
     
Using this observation, we now return to the profile $\vec{p}$ and obtain
\begin{align*}
         \EX_{y \sim \mech(\vec{p}, 1)}[|y|]&= \EX_{y \sim \mech(\vec{p}, 1)}[|y|] + \EX_{y \sim \mech(\vec{p}, 1)}[|y-1|] - \EX_{y \sim \mech(\vec{p}, 1)}[|y-1|]\\
         &=\EX_{y \sim \mech(\vec{p}, 1)}[|y|+ |y-1|] - \EX_{y \sim \mech(\vec{p}, 1)}[|y-1|]\\
         &\geq 1 - \EX_{y \sim \mech(\vec{p}, 1)}[|y-1|]\\
         &\geq 1 - \frac{1}{2}=\frac{1}{2}. \label{eq:cost-ell-lb}\numberthis
\end{align*}
In the above derivation, the second equality follows by the linearity of expectation. Then, the first inequality is due to the triangle inequality. Finally, the second inequality follows from \eqref{eq:deviation-of-r-lb-mb}.
     
Consider now a unilateral deviation of $\ell$ from $p_{\ell}=0$ to $\beta  \in \left(\frac{1-\varepsilon}{2-\varepsilon}, \frac{1}{2}\right)$ (see Figure \ref{fig:MC-LB-RAND-b2}). 
We denote the resulting profile by $\vec{x}':=(\beta ,\vec{p}_{-\ell})$. The following technical claim will be useful for bounding the expected cost of the mechanism.

\begin{claim}\label{claim:fancyclaim}
For every $y \in \mathbb{R}$, $\OPTMC(y, \vec{x}', 1) \geq \frac{\beta }{1-\beta} \cdot |y| + \frac{1/4-\beta}{1-\beta}.$
\end{claim}

\begin{claimproof}
For notation simplicity, let us denote $\lambda:=\frac{\beta }{1-\beta} \in (0,1)$, $\mu:=\frac{\frac{1}{4}- \beta}{1-\beta}$, and $G(y):=\OPTMC(y, \vec{x}', 1)$.
Observe that to prove the claim, it suffices to show that 
\begin{equation}\label{eq:fancyclaimequality}
            G(y) \geq  \lambda |y| + \mu
\end{equation}
holds for every $y \in \mathbb{R}$. Note that for every $y$, either agent $\ell$ at location $\beta$ or agent $r$ at location $1$ is the outlier. In fact, it is easy to verify that the threshold at which the outlier changes is exactly $\frac{1+\beta}{2}$. Hence we can write $G(y)$ as
        \begin{equation*}
           G(y)= \begin{cases}
            \max\left\{ \left| y - \beta \right|,\; \left| y - \frac{1}{2} \right| \right\}, & \text{if } y < \dfrac{1 + \beta}{2}, \\[6pt]
            \max\left\{ \left| y - \frac{1}{2} \right|,\; \left| y - 1 \right| \right\}, & \text{if } y \ge \dfrac{1 + \beta}{2}.
        \end{cases}
        \end{equation*}
We distinguish different cases based on the value of $y$ and show that \eqref{eq:fancyclaimequality} holds in each case.
\medskip
        
\noindent
\underline{Case 1:} $y < \frac{1+ \beta}{2}$. In this case $G(y)= \max\{|y-\beta|, |y- \frac{1}{2}|\}$. We break the analysis into three subcases:

\textit{Case 1a:} $y < 0$. We have
\begin{equation*}
G(y)=\frac{1}{2}-y > -\lambda y + \frac{1}{2}= - \lambda y  + \frac{1}{2} \cdot \frac{1-\beta}{1-\beta} > -\lambda y  + \frac{\frac{1}{4}-\frac{1}{2}\beta}{1-\beta} > - \lambda y +  \frac{\frac{1}{4}-\beta}{1-\beta}= -\lambda y + \mu.
\end{equation*}
The first inequality holds by the definition of the case, while the remaining two inequalities hold since $0<\beta < \frac{1}{2}<1$. The proof of \eqref{eq:fancyclaimequality} follows for this case as $|y|=-y$.

\textit{Case 1b:} $y \in \left[0, \frac{\frac{1}{2}+\beta}{2}\right)$. We have that 
\begin{align*}
G(y)&= \max\{|y-\beta|, |y-\frac{1}{2}|\}\\
&= \frac{1}{2}-y\\
&=\lambda y + \mu +  \frac{1}{2} -(1+\lambda) \cdot y - \mu \\
& > \lambda y + \mu +  \frac{1}{2} -(1+\lambda)\cdot\frac{\frac{1}{2}+\beta}{2} - \mu \\
&= \lambda y+ \mu.
\end{align*}

\textit{Case 1c:} $y \in \left[\frac{\frac{1}{2}+\beta}{2}, \frac{1+\beta}{2}\right)$. 
Note that $\frac{\frac{1}{4}-\beta^2}{1 -\beta} -\beta = \mu$ holds.
Since $y \geq \frac{\frac{1}{2} + \beta}{2}>\beta$, we have that 
        \begin{align*}
            G(y)&= y-\beta\\
            &=\lambda y +\mu -\beta + (1-\lambda) \cdot y - \mu \\
            &\geq \lambda y +\mu -\beta + (1-\lambda) \cdot \frac{\frac{1}{2}+\beta}{2} - \mu\\
            &=\lambda y + \mu -\beta + \frac{\frac{1}{4}-\beta^2}{1 -\beta} - \mu\\
            &= \lambda y + \mu.
        \end{align*}
        
\noindent \underline{Case 2:} $y \geq \frac{1+ \beta}{2}$. In this case $G(y)= \max\{|y-\frac{1}{2}|, |y- 1| \}$. As in Case 1, we break the analysis into multiple subcases.

\textit{Case 2a:} $y \in \left[\frac{1+\beta}{2}, \frac{3}{4} \right)$. Since $y \geq \frac{1+ \beta}{2}$ and $y < 1$ we have that
        \begin{align*}
            G(y)&=1-y\\
            &= \lambda y + \mu +1 - (1+\lambda) \cdot y-\mu\\
            &\geq  \lambda y + \mu +1 - (1+\lambda)\frac{3}{4}-\mu\\
            &= \lambda y + \mu +\frac{1-\beta -1 + \frac{1}{4}-\frac{1}{4}+\beta}{1- \beta} \\
            &= \lambda y + \mu. 
        \end{align*}

\textit{Case 2b:} $y \geq \frac{3}{4}$. Since $y > \frac{1}{2}$ holds, we have that
        \begin{align*}
            G(y)&=y-\frac{1}{2}\\
            &= \lambda y + \mu  -\frac{1}{2} + (1-\lambda)y - \mu \\
            &\geq \lambda y + \mu  -\frac{1}{2} + (1-\lambda)\frac{3}{4} - \mu\\
            &= \lambda y + \mu +\frac{-\frac{1}{2} + \frac{1}{2} \beta + \frac{1}{2} - \beta + \frac
            {1}{4}- \frac
            {1}{2}\beta - \frac
            {1}{4} + \beta}{1 - \beta} \\
            &= \lambda y + \mu.
        \end{align*}
     We have therefore shown \eqref{eq:fancyclaimequality} for every $y \in \mathbb{R}$. This concludes the proof of the claim.
     \end{claimproof}

\noindent We now lower bound the cost of $\mech(\vec{x}', 1)$ as follows:
     \begin{align*}
         \EX_{y \sim \mech(\vec{x}', 1)}\left[\OPTMC(y , \vec{x}', 1)\right] &\geq \EX_{y \sim \mech(\vec{x}', 1)}\left[\frac{\beta}{1-\beta}\cdot |y| + \frac{\frac{1}{4}-\beta}{1-\beta}\right]\\
         &=\frac{\beta}{1-\beta}\cdot \EX_{y \sim \mech(\vec{x}', 1)}\left[ |y| \right]  + \frac{\frac{1}{4}-\beta}{1-\beta}\\
         &\geq \frac{\beta}{1-\beta}\cdot \frac{1}{2} + \frac{\frac{1}{4}-\beta}{1-\beta}\\
         &=\frac{\frac{1}{4}-\frac{\beta}{2}}{1-\beta}
        \numberthis \label{eq:mc-lb-in-exp}.
     \end{align*}
Here the first inequality is due to Claim \ref{claim:fancyclaim} and the subsequent equality follows from the linearity of expectation. Then, the second inequality follows from the strategyproofness of $\mech$ since, because $p_{\ell}=0$ and by \eqref{eq:cost-ell-lb}, we have that $\EX_{y \sim \mech(\vec{x}', 1)}\left[ |y| \right] = \EX_{y \sim \mech(\vec{x}', 1)}\left[ |y-p_{\ell}| \right]\geq \EX_{y \sim \mech(\vec{p}, 1)}\left[ |y-p_{\ell}|\right] =  \EX_{y \sim \mech(\vec{p}, 1)}\left[ |y|\right] \ge \frac{1}{2}$.

We conclude the proof by observing that the optimal location to place the facility under profile $\vec{x}'= (\beta, \frac{1}{2}, 1)$ is the midpoint between $\beta$ and $\frac{1}{2}$, as agent $r$ will be considered an outlier  since, by construction, $\beta \in(0,\frac{1}{2})$. Therefore, using the fact that $\OPTMC^{\star}(\vec{x}',1)= \frac{\frac{1}{2}-\beta}{2}=\frac{1}{4}-\frac{\beta}{2}$, we expand \eqref{eq:mc-lb-in-exp} to obtain 
    
\begin{equation*}
         \EX_{y \sim \mech(\vec{x}', 1)}\left[\OPTMC(y , \vec{x}', 1)\right] \geq \frac{\OPTMC^{\star}(\vec{x}', 1)}{(1-\beta)} > \frac{\OPTMC^{\star}(\vec{x}', 1)}{(1-\frac{1-\varepsilon}{2-\varepsilon})} = (2-\varepsilon) \cdot \OPTMC^{\star}(\vec{x}', 1),
\end{equation*}
which contradicts the assumed approximation guarantee of $\mech$. The lemma follows.    
\end{proof}

Note that the randomized lower bound is in contrast with the setting without outliers. Without outliers, \citet{procaccia2013} give a mechanism that is strategyproof in expectation and achieves an improved expected approximation guarantee of $\frac{3}{2}$.
Our randomized lower bound in Theorem~\ref{lem:MC:LB-2-deterministic} implies that we cannot achieve such an improvement in the setting with outliers.

\section{Minimizing the Utilitarian Objective with Outliers} \label{sec:social-cost}

An optimal solution to the non-strategic problem for the utilitarian objective has a similar structure as the one described for the egalitarian objective: if there are $z$ outliers, an optimal solution disregards $\outliers_{\ell}$ and $\outliers_{r} = \outliers - \outliers_{\ell}$ of the leftmost (smallest) and rightmost (largest) locations of $\vec{x}$, respectively, with $\outliers_{\ell} \in \{0, 1, \ldots, \outliers\}$.
We can therefore identify a set of locations that contains at least one optimal solution for any number of $z$ outliers with $1 \le \outliers \le n-2$, which will turn out to be useful when deriving the approximation guarantees of our mechanisms.

\begin{restatable}{lemma}{fourone} \label{fact:optRange}
Consider an arbitrary profile $\vec{x}$ and $z$ outliers with $1 \le \outliers \le n-2$. Consider the set of locations
\begin{equation} \label{eq:def-O}
O := \Big \{ x_{\sigma(\cl{\frac{n-\outliers +1}{2}})}, x_{\sigma(\cl{\frac{n-\outliers}{2}} +1)}, \dots, x_{\sigma(\cl{\frac{n-\outliers}{2}}) + \outliers}  \Big \}.
\end{equation}
Then, there exists a location $y^{\star}:= y^{\star}(\vec{x}, z) \in O$ which minimizes the utilitarian objective, and $|O| = \outliers$ if $n-\outliers$ is even and $|O| = \outliers +1$ otherwise.
Furthermore, assuming the optimal set of non-outliers $S^{\star}:=S^{\star}(\vec{x},z)$ is with respect to $y^{\star}$, the following statements are true.
\begin{enumerate} 
    \item For every $i \in N \setminus S^{\star}$ with $x_i \le y^{\star}$ and every $j$ with $x_j< x_i$, it holds that $j \in N \setminus S^{\star}$.
    \item For every $i \in N \setminus S^{\star}$ with $x_i \ge y^{\star}$ and every $j$ with $x_j > x_i$, it holds that $j \in N \setminus S^{\star}$.
\end{enumerate}
\end{restatable}

\begin{proof}
For ease of notation, suppose that $x_1 \le x_2 \le \ldots \le x_n$ and for brevity, let $y^{\star}:=y^{\star}(\vec{x}, z)$ and let $S^{\star}:=S^{\star}(\vec{x}, z)$ be the optimal set of non-outliers with respect to $y^{\star}$.
We first show that each optimal solution disregards some of the leftmost and rightmost locations of $\vec{x}$ in the objective.

Consider an $i \in N \setminus S^{\star}$ and assume w.l.o.g. that $x_i \le y^{\star}$.
Towards a contradiction, suppose that there exists a $j \in S^{\star}$ with $x_j < x_i$. 
This contradicts optimality of $(y^{\star}, S^{\star})$, as $(y^{\star}, S')$ with 
$S' = S^{\star} \cup \{ i \} \setminus \{ j \}$ has a smaller egalitarian social cost. 
Namely, for $k \in S^{\star} \cap S'$ the contribution $|x_k - y^{\star}|$ to the social cost remains unchanged.
But for $j \in S^{\star}$ and $i \in S'$, it holds that $|x_j - y^{\star}| > |x_i - y^{\star}|$.
Therefore, given an optimal solution $(y^{\star}, S^{\star})$, there exists a $\outliers_{\ell} \in \{0, 1, \ldots, \outliers\}$ such that $i \notin S^{\star}$ for $i \le \outliers_{\ell}$ and $i > n-(\outliers -\outliers_{\ell})$. By \cite{procaccia2013}, the optimal location $y^{\star}$ is between the $\cl{\frac{n-\outliers}{2}}$-th and $\cl{\frac{n-\outliers +1}{2}}$-th order statistic of the locations corresponding to the $n-\outliers$ agents in $S^{\star}$. 

Specifically, for $n-\outliers$ odd and $\outliers_{\ell} =0$ this leads to $y^{\star} = x_{\cl{\frac{n-\outliers}{2}}} = x_{\cl{\frac{n-\outliers +1}{2}}}$ and as $\outliers_{\ell}$ increases by 1, the order statistic of the location equal to $y^{\star}$ increases by 1 until $\outliers_{\ell} = \outliers$ and $y^{\star} = x_{\cl{\frac{n-\outliers}{2}} + \outliers}$. And so, $|O| = \outliers +1$ in this case. 
For $n-\outliers$ even and $\outliers_{\ell} =0$, this leads to $x_{\cl{\frac{n-\outliers}{2}}} \le y^{\star} \le x_{\cl{\frac{n-\outliers +1}{2}}}$ and as $\outliers_{\ell}$ increases by 1, so do both order statistics of this interval until $\outliers_{\ell} = \outliers$ and $x_{\cl{\frac{n-\outliers}{2}} + \outliers} \le y^{\star} \le x_{\cl{\frac{n-\outliers +1}{2}} + \outliers}$. 
Note that by definition, for $\outliers_{\ell} =0$ only $x_{\cl{\frac{n-\outliers +1}{2}}} \in O$ and for $\outliers_{\ell} =\outliers$ only $x_{\cl{\frac{n-\outliers}{2}} + \outliers} \in O$, leading to $|O| = \outliers$ in this case. 
\end{proof}

As for the problem without outliers, let us consider the mechanism \LM\ that chooses the left median as the location of the facility:
\begin{equation*}\label{eq:median-mech}
\LM(\vec{x}) = x_{\sigma(\cl{\frac{n}{2}})}.
\end{equation*}
Note that if $\outliers = 1$ and $n$ is odd, it follows from Lemma \ref{fact:optRange} that there is always an optimal location $\opt$ in $O$ with $O = \big \{ x_{\sigma(\cl{\frac{n-1 +1}{2}})}, x_{\sigma(\cl{\frac{n-1}{2}} +1)} \big \} = \big \{ x_{\sigma(\cl{\frac{n}{2}})} \big \}$. 
Therefore, for this specific case, \LM\ is 1-approximate, i.e., optimal.
For all other combinations of $n$ and $z$ \LM\ is no longer 1-approximate and as it turns out, for a fixed number of agents $n$ the approximation guarantee increases to $z+1$ if $n$ is even, and $z$ if $n$ is odd, as $z$ increases to $\fl{\frac{n-1}{2}}$.

\begin{restatable}{theorem}{fourtwo} \label{th:SCdeter}
Let $\outliers \in \{1, 2, \ldots, \fl{\frac{n-1}{2}} \}$. Then, \LM\ is strategyproof and has an approximation guarantee for the utilitarian objective of: 
\begin{equation} \label{eq:th:socialcostDeter}
    f(n,\outliers) = \begin{cases}
        \frac{n-1}{n - 2 \outliers + 1}, & \text{if $n$ odd,} \\
        \frac{n}{n - 2 \outliers}, & \text{otherwise.}
    \end{cases} 
\end{equation}
\end{restatable}

We will use the following fact in order to prove the approximation guarantee of Theorem \ref{th:SCdeter}, as it will turn out to be useful to simplify and upper bound the approximation in our analysis. 

\begin{fact} \label{fact:bound}
Let $x \ge y > 0$ and $z \in [0,y)$. Then $\frac{x}{y} \le \frac{x-z}{y -z}$.
\end{fact}

We can now prove Theorem \ref{th:SCdeter}.

\begin{proof}[Proof of Theorem~\ref{th:SCdeter}]
Strategyproofness follows from Theorem \ref{thm:characterization}. 
Let $(\vec{x}, z)$ be an arbitrary profile and let $\alg = \LM(\vec{x})$.
As for any combination of $n$ and $z$ it holds that $f(n,z) \ge 1$, the approximation guarantee follows if there exists an optimal solution $y^{\star}:=y^{\star}(\vec{x}, z)$ such that $\alg = \opt$.
Therefore, suppose that such an optimal solution $\opt$ doesn't exist and let $S^{\star}:=S^{\star}(\vec{x}, z)$ be the optimal set of non-outliers with respect to $y^{\star}$.
Assume w.l.o.g. that $y < \opt$ and define $d = |\opt - \alg|$.
To compare the social cost of $\alg$ and $\opt$, we will evaluate the social cost of $\alg$ with respect to $S^{\star}$.
Note that $\OPTSC(y , \vec{x}, z)= \min_{S \subseteq N:\\|S|=n-z} \sum_{i \in S}|y-x_i| \leq \sum_{i \in S^{\star}}|y-x_i|$ since $|S^*|=n-z$. We have

\begin{align} \label{eq:SC-ub-ratio1}
&\frac{\OPTSC(y,\vec{x},\outliers)}{\OPTSC^{\star}(\vec{x}, \outliers)} 
\le \frac{\sum_{i \in S^{\star}} |\alg - x_i|}{\sum_{i \in S^{\star}} |\opt - x_i|} = \frac{\sum_{i \in S^{\star}: x_i < \alg} |\alg - x_i| + \sum_{i \in S^{\star}: \alg < x_i} |\alg - x_i| }{\sum_{i \in S^{\star}: x_i \le \alg} |\opt - x_i|  + \sum_{i \in S^{\star}: \alg < x_i} |\opt - x_i|  }  \nonumber \\
&\le \frac{\sum_{i \in S^{\star}: x_i < \alg} |\alg - \alg|  + \sum_{i \in S^{\star}: \alg < x_i } |\alg - \opt| }{\sum_{i \in S^{\star}: x_i \le \alg} |\opt - \alg| + \sum_{i \in S^{\star}: \alg < x_i} |\opt - \opt| }  = \frac{\sum_{i \in S^{\star}: \alg < x_i} d }{\sum_{i \in S^{\star}: x_i \le \alg} d  } 
= \frac{ | \{ i \in S^{\star}: \alg < x_i\} | }{ | \{ i \in S^{\star}: x_i \le \alg \}|  }. 
\end{align}
\begin{figure}[b]
\centering
\begin{subfigure}[b]{0.4\linewidth}
    \centering
    \begin{tikzpicture}[scale=0.7]
        \draw[thick] (2,0) -- (10.5,0);
        \filldraw[black] (2.5,0) circle (2pt);
        \filldraw[black] (3,0) circle (2pt);
        \draw [thick, decorate,decoration={brace,amplitude=10pt,mirror},xshift=0.4pt,yshift=-0.4pt](2.35,-0.1) -- (3.15,-0.1) node[black,midway,yshift=-0.6cm] {$\A$};
        
        \draw[thick] (4,0) -- (4,0.5) node[above]{\alg};
        \filldraw[black] (4,0) circle (2pt);
        \draw[stealth-stealth,thick](4.05,0.3) -- (6.95,0.3)  node[color=black, above, xshift=-1cm]{$d$};
        
        \filldraw[black] (5,0) circle (2pt);
        \filldraw[black] (6.5,0) circle (2pt); 
        \draw [thick, decorate,decoration={brace,amplitude=10pt,mirror},xshift=0.4pt,yshift=-0.4pt](4.85,-0.1) -- (6.65,-0.1) node[black,midway,yshift=-0.6cm] {$\B$};
        
        \draw[thick] (7,0) -- (7,0.5) node[above]{\opt};
        \filldraw[black] (7,0) circle (2pt); 
        
        \filldraw[black] (7.5,0) circle (2pt); 
        \filldraw[black] (8,0) circle (2pt); 
        \filldraw[black] (8.5,0) circle (2pt); 
        \filldraw[black] (9.5,0) circle (2pt); 
        \filldraw[black] (10,0) circle (2pt); 
        \draw [thick, decorate,decoration={brace,amplitude=10pt,mirror},xshift=0.4pt,yshift=-0.4pt](7.35,-0.1) -- (10.15,-0.1) node[black,midway,yshift=-0.6cm] {$\C$};
    \end{tikzpicture}
    \caption{ }
    \label{fig:ApproxSocialCostDeterA}
\end{subfigure}
\begin{subfigure}[b]{0.4\linewidth}
    \centering
    \begin{tikzpicture}[scale=0.7]
        \draw[thick] (3.5,0) -- (7.5,0);
    
        \draw[thick] (4,0) -- (4,0.5) node[above]{\alg};
        \filldraw[black] (4,0) circle (2pt);
        \filldraw[black] (4,0) circle (2pt) node[below, yshift=-0.1cm]{$\A$};
        \draw[stealth-stealth,thick](4.05,0.3) -- (6.95,0.3)  node[color=black, above, xshift=-1cm]{$d$};
        
        \draw[thick] (7,0) -- (7,0.5) node[above]{\opt};
        \filldraw[black] (7,0) circle (2pt) node[below, yshift=-0.1cm]{$\B, \C$};
    \end{tikzpicture}
    \vspace{0.3cm}
    \caption{ }
    \label{fig:ApproxSocialCostDeterB}
\end{subfigure}
\caption{Example for bounding the approximation of Theorem \ref{th:SCdeter}, where (b) illustrates how the locations of agents $i \in S^{\star}$ are moved compared to (a) in order to upper bound the guarantee.}
\label{fig:ApproxSocialCostDeter}
\end{figure}
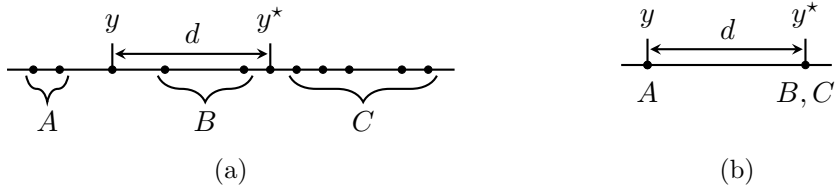
Here, the first equality follows by partitioning the agents $i \in S^{\star}$ depending on their location $x_i$.
As illustrated in Figure \ref{fig:ApproxSocialCostDeter}, the second inequality follows by Fact \ref{fact:bound} by (a) moving locations in $S^{\star}$ with $x_i < \alg$ ($A$ in Figure \ref{fig:ApproxSocialCostDeterA}) to $\alg$ ($A$ in Figure \ref{fig:ApproxSocialCostDeterB}), (b) moving locations in $S^{\star}$ with $\alg < x_i \le \opt$ ($B$ in Figure \ref{fig:ApproxSocialCostDeterA}) to $\opt$ ($B$ in Figure \ref{fig:ApproxSocialCostDeterB}), and (c) moving locations in $S^{\star}$ with $x_i > \opt$ ($C$ in Figure \ref{fig:ApproxSocialCostDeterA}) to $\opt$ ($C$ in Figure \ref{fig:ApproxSocialCostDeterB}).
The last two equalities follow by replacing $|\opt - \alg|$ with $d$ and multiplying both the numerator and denominator by $\frac{1}{d}$.

In order to further upper bound the ratio in \eqref{eq:SC-ub-ratio1}, we want to consider the maximum number of locations $x_i$ with $i \in S^{\star}$ and $\alg < x_i$. 
Using Lemma \ref{fact:optRange}, this number is maximized if $\opt = x_{\sigma(\cl{\frac{n-\outliers}{2}} + \outliers)}$ and $\alg = x_{\sigma(\cl{\frac{n}{2}})} < x_{\sigma(\cl{\frac{n}{2}} +1)}$; note that as $\outliers \le \fl{\frac{n-1}{2}}$, it holds that $\sigma(\cl{\frac{n}{2}}) \in S^{\star}$ by Lemma \ref{fact:optRange}. 
We define $\pi$ as the number of locations $x_i$ with $i \in S^{\star}$ and $\sigma(\cl{\frac{n}{2}}) < \sigma(i) < \sigma(\cl{\frac{n-\outliers}{2}} + \outliers)$.
The exact value of $\pi$ depends on the parities of $\outliers$ and $n$.
Let $i^{\star} = \sigma(\cl{\frac{n-\outliers}{2}} + \outliers)$ and consider the following cases: 

\smallskip
\noindent \underline{Case 1:} $n$ is odd and $\outliers$ is even. 
In this case, $\pi = \frac{\outliers}{2} -1 = \frac{\outliers -2}{2}$. Furthermore, there are $(n-\outliers-1)/2$ agents $i \in S^{\star}$ with $\sigma(i) < i^{\star}$ and $(n-\outliers-1)/2$ agents $i \in S^{\star}$ with $\sigma(i) > i^{\star}$. 
Combining this with \eqref{eq:SC-ub-ratio1} leads to an upper bound of:
\[
\frac{ \frac{n - \outliers -1}{2} + 1 + \frac{\outliers -2}{2} }{ \frac{n - \outliers -1}{2} -  \frac{\outliers -2}{2}}
 = \frac{n  -1 }{n-2\outliers +1}.
\]

\smallskip
\noindent \underline{Case 2:} $n$ is odd and $\outliers$ is odd. Assume that $\outliers \ge 3$, as for $\outliers =1$ and $n$ odd \LM\ is 1-approximate. 
We have that $\pi = \frac{\outliers -1}{2} -1 = \frac{\outliers -3}{2}$. 
Furthermore, there are $((n-\outliers)/2) -1$ agents $i \in S^{\star}$ with $\sigma(i) < i^{\star}$ and $(n-\outliers)/2$ agents $i \in S^{\star}$ with $\sigma(i) > i^{\star}$. 
Combining this with \eqref{eq:SC-ub-ratio1} leads to an upper bound of:
\[
\frac{\frac{n-\outliers}{2}  + 1 + \frac{\outliers -3}{2}}{ \frac{n-\outliers}{2} -1 - \frac{\outliers -3}{2}}
 = \frac{n  -1 }{n-2\outliers +1}.
\]

\smallskip
\noindent \underline{Case 3:} $n$ is even and $\outliers$ is even. In this case, $\pi = \frac{\outliers}{2} -1 = \frac{\outliers -2}{2}$. 
Furthermore, there are $((n-\outliers)/2) -1$ agents $i \in S^{\star}$ with $\sigma(i) < i^{\star}$ and $(n-\outliers)/2$ agents $i \in S^{\star}$ with $\sigma(i) > i^{\star}$. 
Combining this with \eqref{eq:SC-ub-ratio1} leads to an upper bound of:
\[
\frac{\frac{n-\outliers}{2}  + 1 + \frac{\outliers -2}{2}}{ \frac{n-\outliers}{2} -1 - \frac{\outliers -2}{2}}
 = \frac{n}{n-2\outliers }.
\]

\smallskip
\noindent \underline{Case 4:} $n$ is even and $\outliers$ is odd. In this case, $\pi = \frac{\outliers +1}{2} -1 = \frac{\outliers -1}{2}$. 
Furthermore, there are $(n-\outliers-1)/2$ agents $i \in S^{\star}$ with $\sigma(i) < i^{\star}$ and $(n-\outliers-1)/2$ agents $i \in S^{\star}$ with $\sigma(i) > i^{\star}$. 
Combining this with \eqref{eq:SC-ub-ratio1} leads to an upper bound of:
\[
\frac{ \frac{n - \outliers -1}{2} + 1 +  \frac{\outliers -1}{2} }{ \frac{n - \outliers -1}{2} -  \frac{\outliers -1}{2}}
 = \frac{n}{n-2\outliers }.
\]
This concludes the proof.
\end{proof}

We conclude this section by deriving a lower bound for deterministic strategyproof mechanisms, which matches the upper bound in Theorem~\ref{th:SCdeter} for all values of $n$ and $z$. In the proof, we utilize the characterization of~\cite{Moulin80}.

\begin{restatable}{theorem}{fourthree} \label{lem:SC:LB-deterministic}
Let $\outliers \in \{ 1, 2, \ldots, \fl{\frac{n-1}{2}} \}$. Then, there is no deterministic strategyproof mechanism with an approximation guarantee for the utilitarian objective better than: 
\begin{equation} \label{eq:th:socialcostDeter_LB}
    f(n,\outliers) = \begin{cases}
        \frac{n-1}{n - 2 \outliers + 1}, & \text{if $n$ odd,} \\
        \frac{n}{n - 2 \outliers}, & \text{otherwise.}
    \end{cases} 
\end{equation}
\end{restatable}

\begin{proof}
In order to prove our theorem, we will utilize the characterization of~\cite{Moulin80} and the four profiles depicted in Figure~\ref{fig:SC:LB-z+1_short}. 

Recall that Theorem ~\ref{thm:characterization} states that a deterministic mechanism is strategyproof if and only if there exist $n+1$ real numbers $\alpha_1, \ldots, \alpha_{n+1} \in \mathbb{R} \cup \{-\infty,+\infty\}$ such that for every $(x_1, \ldots, x_n) \in \mathbb{R}^n$, the mechanism returns the median of $(x_1, \ldots, x_n, \alpha_1, \ldots, \alpha_{n+1})$. 

So, assume that a deterministic mechanism picks $n+1$ real numbers $\alpha_1, \ldots, \alpha_{n+1} \in \mathbb{R} \cup \{-\infty,+\infty\}$; in what follows assume that w.l.o.g. that $\alpha_i \leq \alpha_{i+1}$, for $i \in [n]$. 
Since the mechanism is deterministic, it has to use the same numbers for every $(x_1, \ldots, x_n) \in \mathbb{R}^n$. In order to prove our claim, we will consider the following two cases.

First, assume that all $\alpha$-points are either $-\infty$, or $+\infty$, and consider a profile with $x_1=x_2= \ldots = x_n = 0$. 
Observe that the optimal utilitarian social cost for this profile is zero. 
In addition, observe that the median of the $(n+1)$-many $\alpha$-points and the $n$-many $x$-points is an $\alpha$-point. 
Thus the mechanism chooses a location with a positive social cost, while the optimal social cost is zero, leading to an unbounded approximation guarantee.

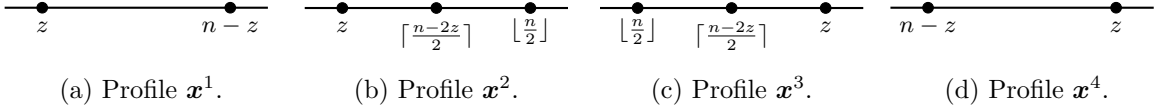
\begin{figure}[t]
\centering
\begin{subfigure}[b]{0.23\linewidth}
    \centering
    \begin{tikzpicture}
        \draw[thick] (0,0) -- (3.5,0);
        \filldraw[black] (0.5,0) circle (2pt) node[below, yshift=-0.05cm]{\footnotesize $\outliers$};
        \filldraw[black] (3,0) circle (2pt) node[below]{\footnotesize $n - \outliers$}; 
    \end{tikzpicture}
    \vspace{0.21cm}
    \caption{Profile $\vec{x}^1$.}
    \label{fig:SC:LB-a}
\end{subfigure}
\begin{subfigure}[b]{0.23\linewidth}
    \centering
    \begin{tikzpicture}
        \draw[thick] (0,0) -- (3.5,0);
        \filldraw[black] (0.5,0) circle (2pt) node[below, yshift=-0.05cm]{\footnotesize $\outliers$};
        \filldraw[black] (3,0) circle (2pt) node[below]{\footnotesize $\lfloor \frac{n}{2} \rfloor$}; 
        \filldraw[black] (1.75,0) circle (2pt) node[below, yshift=-0.05cm]{\footnotesize $\lceil \frac{n-2z}{2} \rceil$}; 
    \end{tikzpicture}
    \caption{Profile $\vec{x}^2$.}
    \label{fig:SC:LB-b}
\end{subfigure}
\begin{subfigure}[b]{0.23\linewidth}
    \centering
    \begin{tikzpicture}
        \draw[thick] (0,0) -- (3.5,0);
        \filldraw[black] (0.5,0) circle (2pt) node[below]{\footnotesize $\lfloor \frac{n}{2} \rfloor$};
        \filldraw[black] (3,0) circle (2pt) node[below, yshift=-0.05cm]{\footnotesize $\outliers$}; 
        \filldraw[black] (1.75,0) circle (2pt) node[below, yshift=-0.05cm]{\footnotesize $\lceil \frac{n-2z}{2} \rceil$}; 
    \end{tikzpicture}
    \caption{Profile $\vec{x}^3$.}
    \label{fig:SC:LB-c}
\end{subfigure}
\begin{subfigure}[b]{0.23\linewidth}
    \centering
    \begin{tikzpicture}
        \draw[thick] (0,0) -- (3.5,0);
        \filldraw[black] (0.5,0) circle (2pt) node[below]{\footnotesize $n - \outliers$};
        \filldraw[black] (3,0) circle (2pt) node[below, yshift=-0.05cm]{\footnotesize $\outliers$}; 
    \end{tikzpicture}
    \vspace{0.21cm}
    \caption{Profile $\vec{x}^4$.}
    \label{fig:SC:LB-d}
\end{subfigure}
\caption{Profiles used in the proof of Theorem \ref{lem:SC:LB-deterministic}.}
\label{fig:SC:LB-z+1_short}
\end{figure}

Now consider the complement of the case above: there exist a $k \in \{0,1, \ldots n-1\}$ such that $\alpha_{n-k} < \alpha_{n-k+1}$. 
We will ``embed'' the four profiles from Figure~\ref{fig:SC:LB-z+1_short} in the interval $(\alpha_{n-k}, \alpha_{n-k+1})$ and show that the mechanism will achieve the claimed bound in one of the profiles. 
This means that for all four profiles, we will assume that $x_i \in (\alpha_k, \alpha_{k+1})$ for every $i \in [n]$.
Observe that in this case, the median of the $\alpha$-points and the $x$-points will be $(k+1)$-th order statistic of $(x_1, x_2, \ldots, x_n)$, i.e., it will be the $(k+1)$-th smallest $x$-point. In addition, observe the following.
\begin{itemize}
    \item If $k +1 \leq z$, then the mechanism achieves an unbounded approximation guarantee for the profile $\vec{x}^1$ depicted in Figure \ref{fig:SC:LB-a}. This is because the optimal solution locates the facility on the right, where $n-z$ agents are located, and achieves a social cost of zero, while the mechanism will locate it on the left leading to a positive social cost.
    \item If $k \geq n-z $, then the mechanism achieves an unbounded approximation guarantee for the profile $\vec{x}^4$ depicted in Figure \ref{fig:SC:LB-d}. This is because the optimal solution locates the facility on the left, where $n-z$ agents are located, and achieves a social cost of zero, while the mechanism will locate it on the right leading to a positive social cost.
\end{itemize}

It remains to consider the cases where $k \in [z,n-z)$.
Assume that $k < n/2$ and consider the profile $\vec{x}^2$ depicted in Figure \ref{fig:SC:LB-b}; if $k \geq n/2$, then we will consider the profile $\vec{x}^3$ depicted in Figure \ref{fig:SC:LB-c} and use symmetric arguments.
Let $d$ denote the distance between the right point, where $ \lfloor \frac{n}{2} \rfloor$ agents are located, and the middle point, where $\lceil \frac{n-2z}{2} \rceil$ agents are located.
Observe the following two facts.
\begin{itemize}
    \item The optimal solution in this case is to locate the facility on the right where $\lfloor \frac{n}{2} \rfloor$ agents have their position, where the cost comes from the agents located in the middle, hence the optimal social cost is $\lceil \frac{n-2z}{2} \rceil \cdot d$.
    \item The $(k+1)$-th order statistic will place the facility on the middle. Hence the social cost of the mechanism will be $\lfloor \frac{n}{2} \rfloor\cdot d$. 
\end{itemize}
Combining the two facts above, we get that the approximation guarantee of the mechanism is $\frac{\lfloor \frac{n}{2}\rfloor}{\lceil \frac{n-2z}{2} \rceil}$, which is $\frac{n}{n-2z}$ if $n$ is even and $\frac{n-1}{n-2z+1}$ if $n$ is odd.
So, overall, we have proven that the $(k+1)$-th order statistic mechanism cannot achieve an approximation guarantee better than the claimed bounds.
\end{proof}

\subsection{Randomized Mechanisms} 

In contrast to the egalitarian objective, positive results can be obtained for the utilitarian objective when allowing for randomization. 
For a profile $\vec{x}$ of even size, our mechanism \RandMed\ randomizes between the left and right median, i.e., $\RandMed(\vec{x})$ returns either $x_{\sigma(\frac{n}{2})}$ or $x_{\sigma(\frac{n}{2} +1)}$, both with probability $\frac{1}{2}$. 

\begin{restatable}{theorem}{sixone} \label{lem:SC:rand:UB}
Let $n$ be even and let $\outliers \in \{ 1, 2, \ldots, \fl{\frac{n-1}{2}} \}$.
Then, \RandMed\ is universally strategyproof and achieves an expected approximation guarantee for the utilitarian objective of:
\begin{equation} \label{th:eq:SC-UB-rand-n-even}
f(n,\outliers) = \begin{cases}
   \frac{n - 1}{n - 2}, & \text{if $\outliers =1$,} \\
    \frac{n^2 - 2nz + 2z}{(n - 2z)(n - 2z +2)}, & \text{otherwise.}
    \end{cases} 
\end{equation}
\end{restatable}

\begin{proof}
Note that \RandMed\ is universally strategyproof as both locations that are potentially returned by \RandMed\ are the outcome of a deterministic strategyproof mechanism (Theorem \ref{thm:characterization}). 
Consider an arbitrary profile $(\vec{x},z)$ with $n$ even and let $y$ be the randomized median, i.e., $y \sim \RandMed(\vec{x}, z)$.

First consider the case of $\outliers =1$.
By Lemma \ref{fact:optRange}, there exists an optimal location $y^{\star} \in \{x_{\sigma(\frac{n}{2})} , x_{\sigma(\frac{n}{2} +1)}\}$, and therefore
\begin{equation*}
        \mathbb{E}[\OPTSC(y, \vec{x}, z)] \le \frac{1}{2} \cdot 1 + \frac{1}{2} \cdot \frac{n}{n - 2 \outliers} = \frac{n - \outliers}{n - 2 \outliers}=\frac{n-1}{n-2},
\end{equation*}
where the first inequality follows by Theorem \ref{th:SCdeter}. 

Now consider the case of $\outliers \ge 2$. 
Assume w.l.o.g. that $y^* \ge x_{\sigma(\frac{n}{2} +1)}$. 
For the left median, i.e., $x_{\sigma(\frac{n}{2})}$, we can use the approximation guarantee of $\frac{n}{n - 2 \outliers}$ as in Theorem \ref{th:SCdeter}. 
For the right median, we can derive an upper bound on the approximation guarantee by using the same approach as in Theorem \ref{th:SCdeter}. 
Namely, if $\outliers$ is even, this leads to an upper bound of
\[
\frac{\frac{n-\outliers}{2}  + 1 + \frac{\outliers -2}{2} -1}{ \frac{n-\outliers}{2} -1 - \frac{\outliers -2}{2} +1}
 = \frac{n -2}{n-2\outliers +2 },
\]
and if $\outliers$ is odd, this leads to an upper bound of
\[
\frac{ \frac{n - \outliers -1}{2} + 1 +  \frac{\outliers -1}{2} -1 }{ \frac{n - \outliers -1}{2} -  \frac{\outliers -1}{2} +1}
 = \frac{n -2}{n-2\outliers +2}.
\]
Finally, this leads to
\begin{equation*}
\mathbb{E}[\OPTSC(y, \vec{x}, z)] \le \frac{1}{2} \cdot \frac{n-2}{n - 2z + 2} + \frac{1}{2} \cdot \frac{n}{n - 2 \outliers} = \frac{n^2 - 2nz + 2z}{(n - 2z)(n - 2z +2)},
\end{equation*}
concluding the proof.
\end{proof}

Recall that universal strategyproofness, the notion achieved by \RandMed, is a stronger requirement than strategyproofness in expectation. 
In Theorems \ref{thm:SC:LB-3over2-randomized} and \ref{thm:SC:LB-2-randomized-odd} below, we derive two matching lower bounds for $(n,z)=(4,1)$ and $(n,z)=(5,2)$ for randomized mechanisms that achieve the weaker notion of strategyproofness in expectation. 
Interestingly, the latter lower bound matches the upper bound achieved by $\Lmedian$ in Theorem \ref{th:SCdeter} for $n=5$ and $z=2$. Understanding whether randomization helps improve the approximation guarantee for odd $n$ is an interesting question for follow-up work. 

The following lemma will be useful when proving the aforementioned negative results. 

\begin{lemma}[see also \citep{procaccia2013}]\label{lem:triangle-inequality}
    Let $D$ be a probability distribution supported on $\mathbb{R}$. Then, for every $x_1, x_2 \in \mathbb{R}$, it holds that
    \[
        \max\left\{ \EX_{y \sim D}[|y - x_1|],\; \EX_{y \sim D}[|y - x_2|] \right\} \geq \frac{|x_1 - x_2|}{2}.
    \]
\end{lemma}

\begin{proof}
We observe that
\begin{align*}
    \max\left\{ \EX_{y \sim D}[|y - x_1|],\; \EX_{y \sim D}[|y - x_2|] \right\} 
    &\geq \EX_{y \sim D}\left[ \frac{ |y - x_1| + |y - x_2| }{2} \right]
    \geq 
    \frac{|x_1-x_2|}{2},
\end{align*}
where the first and second inequality follow by linearity of expectation and the triangle inequality, respectively.
\end{proof}

We first prove the negative results for $(n,z)=(4,1)$.
\medskip

\noindent

\begin{theorem} \label{thm:SC:LB-3over2-randomized}
Let $n=4$ and $\outliers = 1$. Then, there is no randomized mechanism that is strategyproof in expectation and achieves an expected approximation guarantee better than $\nicefrac{3}{2}$ for the utilitarian objective.  
\end{theorem} 

\begin{proof}
Let $\mathcal{M}$ be any randomized mechanism that is strategyproof in expectation. Consider the profile with $N = \{\ell, m_1, m_2, r\}$ and preferred locations $p_{\ell} = 0$, $p_{m_1} = \nicefrac{1}{3}$, $p_{m_2} = \nicefrac{2}{3}$, and $p_r = 1$. Suppose the agents report truthfully, i.e., $\vec{x} = \vec{p}$. By applying Lemma~\ref{lem:triangle-inequality} to $\mech(\vec{x}, 1)$ and the points $x_{\ell}$ and $x_{r}$, we conclude that there exists an agent $i \in \{\ell, r\}$ such that
\begin{equation}\label{eq:sc-triangle-lb}
    \EX_{y \sim \mech(\vec{x},1)}\left[ |y - x_i| \right] \geq \frac{|x_{\ell} - x_{r}|}{2} = \frac{1}{2}.
\end{equation}

We distinguish two cases based on whether the agent $i$ satisfying \eqref{eq:sc-triangle-lb} is $\ell$ or $r$. Suppose $i = r$, and consider a unilateral deviation by $r$ to $x'_{r} = 2/3$ (see Figure~\ref{fig:SC-LB-RAND-even}). Let $\vec{x}' = (x_{\ell}, x_{m_1}, x_{m_2}, x'_r)$ be the new profile after the deviation. Since $\mathcal{M}$ is strategyproof in expectation, and using \eqref{eq:sc-triangle-lb}, we have:
\begin{equation}\label{eq:sc-lb-sp-lb}
    \EX_{y \sim \mech(\vec{x}',1)}\left[|y - 1|\right] = \EX_{y \sim \mech(\vec{x}',1)}\left[|y - p_{r}|\right] 
    \geq \EX_{y \sim \mech(\vec{x},1)}\left[|y - p_{r}|\right] 
    = \EX_{y \sim \mech(\vec{x},1)}\left[|y - 1|\right] 
    \geq \frac{1}{2}.
\end{equation}

To proceed, we state and prove the following technical claim.

\begin{claim}\label{claim:cases-y-sc-lb}
    For every $y \in \mathbb{R}$, it holds that $\OPTSC(y, (\vec{x}',1)) \geq |y - 1|$.
\end{claim}

\begin{claimproof}
Note that for every $y \in \mathbb{R}$, the outlier for $\vec{x}'$ is, by construction, either the leftmost agent $\ell$ at $0$ or one of the two rightmost agents at $\frac{2}{3}$. In fact, the threshold $\tau \in \mathbb{R}$ for a given $y$ at which the outlier changes is such that $|\tau|=|\tau-\frac{2}{3}|$, or equivalently $\tau=\frac{1}{3}$. We can therefore write $G(y)$ as
\begin{equation*}
   G(y)= \begin{cases}
    |y| + |y-\frac{1}{3}| + |y-\frac{2}{3}|,  & \text{if } y < \tau, \\[6pt]
    |y-\frac{1}{3}| + 2|y-\frac{2}{3}|, & \text{if } y \ge \tau.
\end{cases}
\end{equation*}

We now show that $G(y) \geq |y-1|$ holds for every $y \in \mathbb{R}$. When $y < \tau=\frac{1}{3}$, we can use the triangle inequality to obtain that
\begin{equation*}
    G(y) = |y| + \Big |y - \frac{1}{3} \Big | + \Big |y-\frac{2}{3} \Big | \geq |y| +  |2y-1| \geq |y-1|.
\end{equation*}
Similarly, when $y \geq \tau=\frac{1}{3}$, we have that
\begin{equation*}
    G(y)= \Big | y-\frac{1}{3} \Big | + 2 \Big | y-\frac{2}{3} \Big | \geq \Big |y-\frac{1}{3} \Big | + \Big | 2y-\frac{4}{3} \Big |  \geq |y-1|.
\end{equation*}
The claim follows.
\end{claimproof}

We conclude that
\begin{align*}
    \EX_{y \sim \mech(\vec{x}',1)}\left[\OPTSC(y, (\vec{x}',1)) \right] 
    &\geq \EX_{y \sim \mech(\vec{x}',1)}\left[|y - 1| \right]  \geq \frac{1}{2},
\end{align*}
where the first inequality follows from Claim~\ref{claim:cases-y-sc-lb} and the final inequality follows from \eqref{eq:sc-lb-sp-lb}. 
This completes the proof for the case of $i = r$, since the optimal social cost for the profile $\vec{x}'$ with $z = 1$ is obtained by excluding agent $\ell$ and placing the facility at $x_{r}'=2/3$, i.e., $\OPTSC^{\star}(\vec{x}',1) = 1/3$.

The case $i = \ell$ is symmetric, by considering a unilateral deviation of agent $\ell$ from $\vec{x}$ to $1/3$.
\end{proof}

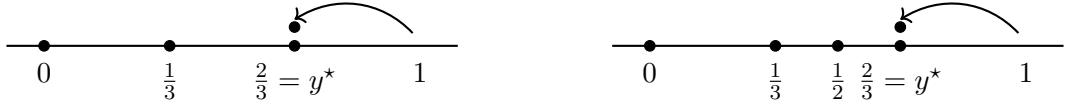
\begin{figure}[h]
\centering
\begin{subfigure}[b]{0.48\linewidth}
    \centering
    \begin{tikzpicture}
        \draw[thick] (-1,0) -- (5,0);
        \filldraw[black] (-0.5,0) circle (2pt) node[below, yshift=-0.1cm]{$0$};
        \filldraw[black] (4.5,0) circle (0.00002pt) node[below, yshift=-0.1cm]{$1$}; 
        \filldraw[black] (1.17,0) circle (2pt) node[below, yshift=-0.1cm]{$\frac{1}{3}$};
        \filldraw[black] (2.83,0) circle (2pt) node[below, yshift=-0.1cm]{$\frac{2}{3}=y^{\star}$}; 
        \filldraw[black] (2.83,0 .25) circle (2pt) node[below, yshift=0.2cm]{$ $};

         \draw[->, thick, bend right=40] (4.4,0.17) to (2.83,0.35);
    \end{tikzpicture}
    \caption{Deviation under $\vec{p}$ for agent $r$ to $\nicefrac{2}{3}$ with $z=1$.}
    \label{fig:SC-LB-RAND-even}
\end{subfigure}
\begin{subfigure}[b]{0.48\linewidth}
    \centering
    \begin{tikzpicture}
        \draw[thick] (-1,0) -- (5,0);
        \filldraw[black] (-0.5,0) circle (2pt) node[below, yshift=-0.1cm]{$0$};
        \filldraw[black] (4.5,0) circle (0.00002pt) node[below, yshift=-0.1cm]{$1$}; 
        \filldraw[black] (1.17,0) circle (2pt) node[below, yshift=-0.1cm]{$\frac{1}{3}$};
        \filldraw[black] (2.83,0) circle (2pt) node[below, yshift=-0.1cm]{$\frac{2}{3}=y^{\star}$};
        \filldraw[black] (2,0) circle (2pt) node[below, yshift=-0.1cm]{$\frac{1}{2}$}; 

        \filldraw[black] (2.83,0 .25) circle (2pt) node[below, yshift=0.2cm]{$ $};

         \draw[->, thick, bend right=40] (4.4,0.17) to (2.83,0.35);
    \end{tikzpicture}
    \caption{Deviation under $\vec{p}$ for agent $r$ to $\nicefrac{2}{3}$, with $z=2$.}
    \label{fig:SC-LB-RAND-odd}
\end{subfigure}
\caption{Profiles used in the proofs of Theorem \ref{thm:SC:LB-3over2-randomized} and Theorem \ref{thm:SC:LB-2-randomized-odd}.}
\label{fig:SC-LB-RAND}
\end{figure}

\begin{theorem} \label{thm:SC:LB-2-randomized-odd}
Let $n=5$ and $\outliers = 2$. Then, there is no randomized mechanism that is strategyproof in expectation and achieves an expected approximation guarantee better than $2$ for the utilitarian objective.  
\end{theorem} 

\begin{proof}
Let $\mathcal{M}$ be any randomized mechanism that is strategyproof in expectation. Consider the profile with $N = \{\ell, m_{1}, m, m_{2}, r\}$ and preferred locations $p_{\ell} = 0$, $p_{m_1} = 1/3$, $p_m=1/2$, $p_{m_2} = 2/3$, and $p_r = 1$, as depicted in Figure~\ref{fig:SC-LB-RAND-odd}. Suppose the agents report truthfully, i.e., $\vec{x} = \vec{p}$. By applying Lemma~\ref{lem:triangle-inequality} to $\mech(\vec{x}, 2)$ and the points $x_{\ell}$ and $x_{r}$, we conclude that there exists an agent $i \in \{\ell, r\}$ such that
\begin{equation}\label{eq:sc-triangle-lb-odd}
    \EX_{y \sim \mech(\vec{x},2)}\left[ |y - x_i| \right] \geq \frac{|x_{\ell} - x_{r}|}{2} = \frac{1}{2}.
\end{equation}

We distinguish two cases based on whether the agent $i$ satisfying \eqref{eq:sc-triangle-lb-odd} is $\ell$ or $r$. Suppose $i = r$, and consider a unilateral deviation by $r$ to $x'_{r} = 2/3$ (see Figure~\ref{fig:SC-LB-RAND-odd}). Let $\vec{x}' = (x_{\ell}, x_{m_1}, x_{m}, x_{m_2}, x'_r)$ be the new profile after the deviation. Since $\mathcal{M}$ is strategyproof in expectation, and using \eqref{eq:sc-triangle-lb-odd}, we have
\begin{equation}\label{eq:sc-lb-sp-lb-odd}
    \EX_{y \sim \mech(\vec{x}',2)}\left[|y - 1|\right] = \EX_{y \sim \mech(\vec{x}',2)}\left[|y - p_{r}|\right] 
    \geq \EX_{y \sim \mech(\vec{x},2)}\left[|y - p_{r}|\right] 
    = \EX_{y \sim \mech(\vec{x},2)}\left[|y - 1|\right] 
    \geq \frac{1}{2}.
\end{equation}

To proceed, we state and prove the following technical claim.

\begin{claim}\label{claim:cases-y-sc-lb-odd}
    For every $y \in \mathbb{R}$, it holds that $\OPTSC(y, (\vec{x}',2)) \geq |y - 1|-\frac{1}{6}$.
\end{claim}

\begin{claimproof}
Note that for every $y \in \mathbb{R}$, the two outliers for $\vec{x}'$ are, by construction, the two leftmost agents $\ell$ and $m_1$ at $0$ and $1/3$, the leftmost agents $\ell$ at $0$ and the rightmost agent $r$ at $2/3$, or the two rightmost agents at $\frac{2}{3}$. 
We can therefore write $G(y)$ as
\begin{equation*}
   G(y)= \begin{cases}
    |y| + |y-\frac{1}{3}| + |y-\frac{1}{2}|,  & \text{if } y \le \frac{1}{3},  \\[6pt]
    |y-\frac{1}{3}| + |y-\frac{1}{2}| + |y-\frac{2}{3}|, & \text{if } \frac{1}{3} < y < \frac{1}{2}, \\[6pt]
    |y-\frac{1}{2}| + 2|y-\frac{2}{3}|, & \text{if } y \ge \frac{1}{2}.
\end{cases}
\end{equation*}

We now show that $G(y) \geq |y-1| - \frac{1}{6}$ holds for every $y \in \mathbb{R}$. When $y \le \frac{1}{3}$, we can use the triangle inequality to obtain that
\begin{equation*}
    G(y) = |y| + \Big | y-\frac{1}{3} \Big | + \Big | y-\frac{1}{2} \Big | \geq |y| + \Big |2y-\frac{5}{6} \Big | \geq \Big | y-\frac{5}{6} \Big | \geq |y-1| - 
    \frac{1}{6}.
\end{equation*}
Similarly, when $\frac{1}{3} < y < \frac{1}{2}$, we have that
\begin{equation*}
G(y)= \Big | y-\frac{1}{3} \Big | + \Big | y-\frac{1}{2} \Big | + \Big | y-\frac{2}{3} \Big | \ge \Big | y-\frac{1}{3} \Big | + \Big | 2y-\frac{7}{6} \Big | \ge 
\Big |y-\frac{5}{6} \Big | \geq |y-1|-\frac{1}{6}.
\end{equation*}
Finally, when $y \geq \frac{1}{2}$, we have that
\begin{equation*}
    G(y)= \Big | y-\frac{1}{2} \Big | + 2 \Big | y-\frac{2}{3} \Big | = \Big |y-\frac{1}{2} \Big | + \Big |2y-\frac{4}{3} \Big | \geq \Big |y-\frac{5}{6} \Big | \geq |y-1|-\frac{1}{6}.
\end{equation*}
The claim follows.
\end{claimproof}

We conclude that
\begin{align*}
    \EX_{y \sim \mech(\vec{x}',2)}\left[\OPTSC(y, (\vec{x}',2)) \right] \geq \EX_{y \sim \mech(\vec{x}',2)}\left[|y - 1| \right] -\frac{1}{6}  \geq \frac{1}{3},
\end{align*}
where the first inequality follows from Claim~\ref{claim:cases-y-sc-lb-odd} and by linearity of expectation, and the second inequality follows from \eqref{eq:sc-lb-sp-lb-odd}.
This completes the proof for the case $i = r$, since the optimal social cost for the profile $\vec{x}'$ with $z = 2$ is obtained by excluding agents $\ell$ and $m_1$ and placing the facility at $x_{r}'=2/3$, i.e., $\OPTSC^{\star}(\vec{x}',2) = \nicefrac{1}{6}$.

The case $i = \ell$ is symmetric, by considering a unilateral deviation of agent $\ell$ from $\vec{x}$ to $\nicefrac{1}{3}$.
\end{proof}

\section{Augmenting with Output Predictions} \label{sec:predictions}

In this section, we augment our base model with predictions and provide tight bounds for deterministic strategyproof mechanisms under the consistency-robustness dimension.

\subsection{Impossibility Result for the Egalitarian Objective} \label{sec:predictions-MC}

Augmenting facility location problems with a prediction of the optimal location is common in the literature of learning-augmented mechanism design.  
For the single facility location problem without outliers on the real line, \cite{agrawal2024} obtained a mechanism that is $1$-consistent and $2$-robust, achieving the best of both worlds.
However, when the problem with outliers is augmented with a prediction of the optimal location, achieving a consistency guarantee that is better than the optimal worst-case guarantee of $2$ (Theorem \ref{thm:MC:sp-2approx}) inevitably leads to an unbounded robustness guarantee, as shown in Theorem \ref{lem:MC:LB-prediction-opt} below.

\begin{restatable}{theorem}{fiveone} \label{lem:MC:LB-prediction-opt}
Let $\outliers \le \fl{\frac{n-1}{2}}$.
Then, there is no deterministic strategyproof mechanism augmented with a prediction $\hat{y}$ of the optimal location that is better than $2$-consistent and achieves a finite robustness for the egalitarian objective.
\end{restatable} 

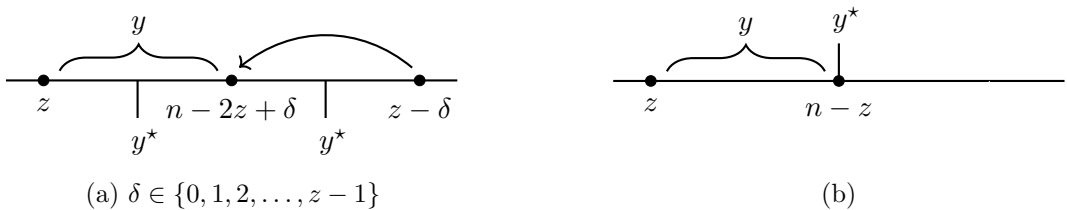
\begin{figure}[b]
\centering
\begin{subfigure}[b]{0.48\linewidth}
    \centering
    \begin{tikzpicture}
        \draw[thick] (-1,0) -- (5,0);
        \filldraw[black] (-0.5,0) circle (2pt) node[below, yshift=-0.1cm]{$\outliers$};
        \filldraw[black] (4.5,0) circle (2pt) node[below, yshift=-0.1cm]{$\outliers - \delta$}; 
        \filldraw[black] (2,0) circle (2pt) node[below, yshift=-0.1cm]{$n - 2 \outliers + \delta$}; 
        \draw[thick] (0.75,0) -- (0.75,-0.5) node[below, xshift=0.1cm]{$y^{\star}$};
        \draw[thick] (3.25,0) -- (3.25,-0.5) node[below, xshift=0.1cm]{$y^{\star}$};
        \draw [thick, decorate,decoration={brace,amplitude=10pt},xshift=0.4pt,yshift=0.8pt](-0.3,0.1) -- (1.8,0.1) node[black,midway,yshift=0.6cm] {$\alg$};
         \draw[->, thick, bend right=40] (4.4,0.17) to (2.1,0.17);
    \end{tikzpicture}
    \caption{$\delta \in \{0, 1, 2, \ldots, \outliers -1\}$ }
    \label{fig:MC-LB-2-b2}
\end{subfigure}
\begin{subfigure}[b]{0.48\linewidth}
    \centering
    \begin{tikzpicture}
        \draw[thick] (-1,0) -- (5,0);
        \filldraw[black] (-0.5,0) circle (2pt) node[below, yshift=-0.1cm]{$\outliers$};
        \filldraw[black] (2,0) circle (2pt) node[below, yshift=-0.1cm]{$n - \outliers$}; 
        \draw[thick] (2,0) -- (2,0.5) node[above, xshift=0.1cm]{$y^{\star}$};
        \draw[color=white] (4,0) -- (4,-0.5) node[color = white, below, xshift=0.1cm]{$y^{\star}$}; 
        \draw [thick, decorate,decoration={brace,amplitude=10pt},xshift=0.4pt,yshift=0.8pt](-0.3,0.1) -- (1.8,0.1) node[black,midway,yshift=0.6cm] {$\alg$};
    \end{tikzpicture}
    \caption{ }
    \label{fig:MC-LB-2-c2}
\end{subfigure}
\caption{Profiles used in the proof of Theorem \ref{lem:MC:LB-prediction-opt}.}
\label{fig:MC-LB-22}
\end{figure}

\begin{proof}
Let $\outliers \le \fl{\frac{n-1}{2}}$.
Towards a contradiction, assume that there exists a strategyproof mechanism $\mathcal{M}$ augmented with $\hat{y}$ that is $(2 - \varepsilon)$-consistent, with $\varepsilon>0$, and $\beta$-robust with $\beta < \infty$. 
Consider the profiles depicted in Figure \ref{fig:MC-LB-22}. 

The profile in Figure \ref{fig:MC-LB-2-b2} with $\delta =0$ has 3 clusters of locations: there are $z$ agents with a location of 0, $n - 2 \outliers \ge 1$ agents with a location of $\frac{1}{2}$ and $z$ agents with a location of 1.
Note that all locations from either the leftmost cluster at 0 or the rightmost cluster at 1 can be disregarded in the objective function. 
Therefore, there are two optimal locations $y^{\star}$ at $\frac{1}{4}$ and $\frac{3}{4}$ with an egalitarian social cost of $\frac{1}{4}$. 
Consider the case of a perfect prediction, i.e., $\hat{y} = \frac{1}{4}$ or $\hat{y} = \frac{3}{4}$. 
In either case, as $\mathcal{M}$ is $(2 - \varepsilon)$-consistent, it must be that $\mathcal{M}$ places the facility $\alg$ in $[\frac{\varepsilon}{4}, \frac{1}{2} - \frac{\varepsilon}{4}]$ or in $[\frac{1}{2} + \frac{\varepsilon}{4}, 1 - \frac{\varepsilon}{4}]$. Assume w.l.o.g. that $\alg \in [\frac{\varepsilon}{4}, \frac{1}{2} - \frac{\varepsilon}{4}]$, depicted by the curly bracket in Figure \ref{fig:MC-LB-2-b2}. 

Now, consider the profile in Figure \ref{fig:MC-LB-2-b2} with $\delta =1$: one agent with a location of 1 moved to $\frac{1}{2}$.
Note that as $\mathcal{M}$ is strategyproof, the location $\alg$ remains unchanged by Corollary \ref{lem:SP-property-M}.
The same reasoning holds when considering $\delta = 2, 3, \ldots, \outliers -1$ consecutively.

Finally, consider the profile in Figure \ref{fig:MC-LB-2-c2} in which all $\outliers$ agents with a location of 1 have moved to $\frac{1}{2}$. 
As the left cluster now contains $\outliers$ locations, all locations from this cluster can be disregarded, leading to a unique optimal location $y^{\star}$ at $\frac{1}{2}$ with an egalitarian social cost of 0.
However, as $\mathcal{M}$ is strategyproof the location $\alg \in [\frac{\varepsilon}{4}, \frac{1}{2} - \frac{\varepsilon}{4}]$ still remains unchanged by Corollary \ref{lem:SP-property-M}, and so $\alg$ has a positive cost.
Therefore, this contradicts robustness if $\hat{y} = \frac{3}{4}$ and consistency if $\hat{y} = \frac{1}{4}$, concluding the proof. 
\end{proof}

Additionally, augmenting the problem with other natural predictions leads to the same negative result.  
One example is a prediction $\hat{\outliers}_{\ell}$ that indicates how many of the leftmost (smallest) $x_i$'s are outliers (which implies how many of the rightmost (largest) $x_i$'s are outliers, i.e., $\hat{\outliers}_{r} = \outliers - \hat{\outliers}_{\ell}$). 
Another example is a prediction indicating which $\outliers$ agents are outliers. 

\subsection{Positive Results for the Utilitarian Objective} \label{sec:predictions-SC}

Without predictions, we have shown in Lemma \ref{lem:approx-unbounded-z>=halfn} that no deterministic strategyproof mechanism exists if $z \ge \nicefrac{n}{2}$. 
When augmenting the problem with a prediction $\hat{y}$ of the optimal location, this inapproximability persists for the robustness guarantee. 
In fact, if the number of outliers is \textit{greater than one-third} of the total number of agents, no deterministic strategyproof mechanism augmented with a prediction $\hat{y}$ can be $1$-consistent and achieve a finite robustness.
This follows from the more general lemma below, which considers a relaxed consistency guarantee of $\frac{\outliers}{n - 2\outliers} > 1$ for $\outliers > 1 $ and $n \in \{2\outliers+1, 2\outliers+2, \ldots, 3\outliers-1 \}$.

\begin{restatable}{theorem}{fivetwo}  \label{lem:SC:LB-prediction-z>1-unbounded}
Let $\outliers > 1 $ and $n \in \{2\outliers+1, 2\outliers+2, \ldots, 3\outliers-1 \}$.
Then, there is no deterministic strategyproof mechanism augmented with a prediction $\hat{y}$ of the optimal location that is better than $\frac{\outliers}{n - 2\outliers}$-consistent and achieves a finite robustness for the utilitarian objective.
\end{restatable} 

\begin{figure}[b]
\centering
\begin{subfigure}[b]{0.4\linewidth}
    \centering
    \begin{tikzpicture}
        \draw[thick] (-0.5,0) -- (4.5,0);
        \filldraw[black] (0,0) circle (2pt) node[below, yshift=-0.1cm]{$\outliers - \delta$};
        \filldraw[black] (2.6,0) circle (2pt) node[below, yshift=-0.1cm]{$n - 2 \outliers + \delta$};
        \filldraw[black] (3.9,0) circle (2pt) node[below, yshift=-0.1cm]{$\outliers$}; 
        \draw[thick] (3.9,0) -- (3.9,0.5) node[above, xshift=0.0cm]{$\hat{y}$};
        \draw[->, thick, bend left=40] (0.1,0.17) to (2.5,0.17);
    \end{tikzpicture}
    \caption{$\delta \in \{0, 1, 2, \ldots, \outliers -1\}$ }
    \label{fig:SC-LB-pred-z>1-unbounded-b}
\end{subfigure}
\begin{subfigure}[b]{0.4\linewidth}
    \centering
    \begin{tikzpicture}
        \draw[thick] (-0.5,0) -- (4.5,0);
        \filldraw[black] (2.6,0) circle (2pt) node[below, yshift=-0.1cm]{$n - \outliers$};
        \filldraw[black] (3.9,0) circle (2pt) node[below, yshift=-0.1cm]{$\outliers$}; 
        \draw[thick] (2.6,0) -- (2.6,0.5) node[above, xshift=0.1cm]{$\opt$};
        \draw[thick] (3.9,0) -- (3.9,0.5) node[above, xshift=0.0cm]{$\hat{y}$};
    \end{tikzpicture}
    \vspace{0.05cm}
    \caption{ }
    \label{fig:SC-LB-pred-z>1-unbounded-c}
\end{subfigure}
\caption{Profiles used in the proof of Theorem \ref{lem:SC:LB-prediction-z>1-unbounded}.}
\label{fig:SC-LB-pred-z>1-unbounded}
\end{figure}
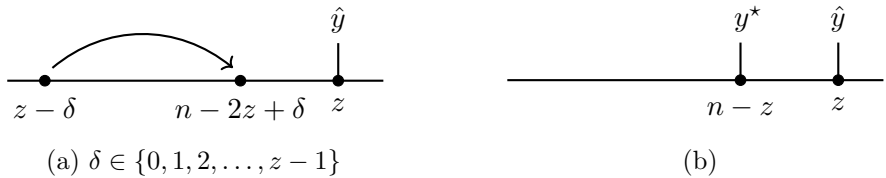

\begin{proof}
Towards a contradiction, assume that there exists a strategyproof mechanism $\mathcal{M}$ that is better than $\frac{\outliers}{n - 2\outliers}$-consistent and achieves a finite robustness.
Consider the profiles depicted in Figure \ref{fig:SC-LB-pred-z>1-unbounded}.

The profile in Figure \ref{fig:SC-LB-pred-z>1-unbounded-b} with $\delta =0$ has 3 clusters of locations: there are $z$ agents with a location of 0, $n - 2 \outliers \ge 1$ agents with a location of $z \cdot d >0$ and $z$ agents with a location of $(z +1)d$. 
As $n < 3\outliers$, there is one optimal location $\opt$ at $(z +1)d$ with a utilitarian social cost of $(n-2\outliers)d$. 
Consider a perfect prediction $\hat{y}$.
As $\mathcal{M}$ is better than $\frac{\outliers}{n - 2\outliers}$-consistent, it must be that $\mathcal{M}$ places the facility at $y$ such that the social cost is smaller than $z \cdot d$.
To achieve this, $\mathcal{M}$ must place the facility at $y$ such that $\alg > z \cdot d$.

Now, consider the profile in Figure \ref{fig:SC-LB-pred-z>1-unbounded-b} with $\delta =1$ in which one agent with a location of 0 moved to $z \cdot d$. 
As $\mathcal{M}$ is strategyproof, it must be that $\alg$ remains unchanged by Corollary \ref{lem:SP-property-M}.
The same reasoning holds when considering $\delta = 2, 3, \ldots, \outliers -1$ consecutively. 

Finally, consider the profile in Figure \ref{fig:SC-LB-pred-z>1-unbounded-c} in which all $z$ agents with a location of 0 moved to $z \cdot d$.
In this case, there is one optimal location $\opt$ at $z \cdot d$ with a utilitarian social cost of $0$. 
So in order for $\mathcal{M}$ to achieve a bounded robustness, note that $\hat{y} \neq \opt$, it must place the facility at $\opt$.
However, this contradicts strategyproofness of $\mathcal{M}$ by Corollary \ref{lem:SP-property-M} as $y > z \cdot d$, concluding the proof. 
\end{proof}

Note that for $\outliers > 1$ and $n \in \{ \outliers +2, \outliers +3, \ldots, 2\outliers \}$, it follows from Theorem \ref{lem:approx-unbounded-z>=halfn} that there is no deterministic strategyproof mechanism augmented with a prediction $\hat{y}$ that is $1$-consistent and achieves a finite robustness.
Finally, if the number of outliers is at most one-third of the total number of agents, we show that no deterministic strategyproof mechanism augmented with a prediction $\hat{y}$ can be $1$-consistent and achieve a specific robustness depending on the number of agents and outliers. 

\begin{restatable}{theorem}{fivethree}  \label{lem:SC:LB-prediction-z>1-bounded}
Let $n \ge 3 \outliers$.
Then, there is no deterministic strategyproof mechanism augmented with a prediction $\hat{y}$ of the optimal location that is $1$-consistent and better than $f(n,\outliers)$-robust for the utilitarian objective with:
\begin{equation} \label{eq:SC-LB-robustness}
 f(n,\outliers) = \begin{cases}
        \frac{n+\outliers-1}{n-3\outliers+1}, & \text{if $n-\outliers$ is odd,} \\
        \frac{n+\outliers-2}{n-3\outliers+2}, & \text{otherwise.}
    \end{cases} 
\end{equation}
\end{restatable} 

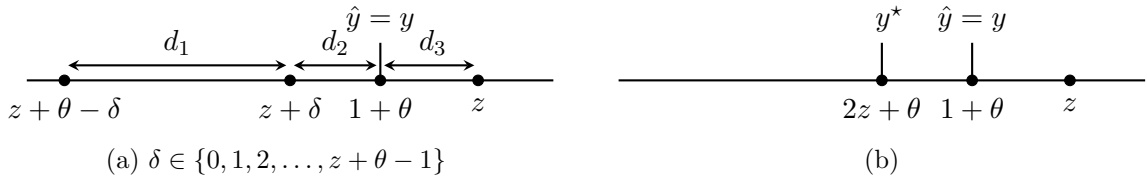
\begin{figure}[b]
\centering
\begin{subfigure}[b]{0.48\linewidth}
    \centering
    \begin{tikzpicture}
        \draw[thick] (-0.5,0) -- (6.5,0);
        \filldraw[black] (0,0) circle (2pt) node[below, yshift=-0.1cm]{$\outliers + \theta - \delta$};
        \filldraw[black] (3,0) circle (2pt) node[below, yshift=-0.1cm]{$\outliers + \delta$};
        \filldraw[black] (4.2,0) circle (2pt) node[below, yshift=-0.1cm]{$1 + \theta$}; 
        \filldraw[black] (5.5,0) circle (2pt) node[below, yshift=-0.1cm]{$\outliers$}; 
        \draw[thick] (4.2,0) -- (4.2,0.5) node[above, xshift=0cm]{$\hat{y} = \alg$};
        \draw[stealth-stealth,thick](0.05,0.2) -- (2.95,0.2)  node[color=black, above, xshift=-1.45cm]{$d_1$};
        \draw[stealth-stealth,thick](3.05,0.2) -- (4.15,0.2)  node[color=black, above, xshift=-0.55cm]{$d_2$};
        \draw[stealth-stealth,thick](4.25,0.2) -- (5.45,0.2)  node[color=black, above, xshift=-0.55cm]{$d_3$};
    \end{tikzpicture}
    \caption{$\delta \in \{0,1, 2, \ldots, \outliers + \theta -1\}$}
    \label{fig:SC-LB-pred-z>1-unbounded-n-big-a}
\end{subfigure}
\begin{subfigure}[b]{0.48\linewidth}
    \centering
        \begin{tikzpicture}
        \draw[thick] (-0.5,0) -- (6.5,0);
        \filldraw[black] (3,0) circle (2pt) node[below, yshift=-0.1cm]{$2\outliers + \theta$};
        \filldraw[black] (4.2,0) circle (2pt) node[below, yshift=-0.1cm]{$1 + \theta$}; 
        \filldraw[black] (5.5,0) circle (2pt) node[below, yshift=-0.1cm]{$\outliers$}; 
        \draw[thick] (4.2,0) -- (4.2,0.5) node[above, xshift=0cm]{$\hat{y} = \alg$};
        \draw[thick] (3,0) -- (3,0.5) node[above, xshift=0.1cm]{$\opt$};
    \end{tikzpicture}
    \caption{ }
    \label{fig:SC-LB-pred-z>1-unbounded-n-big-b}
\end{subfigure}
\caption{Profiles with $n-z$ odd used in the proof of Theorem \ref{lem:SC:LB-prediction-z>1-bounded}.}
\label{fig:SC-LB-pred-z>1-unbounded-n-big}
\end{figure}

\begin{proof}
Towards a contradiction, assume that there exists a strategyproof mechanism $\mathcal{M}$ augmented with a prediction $\hat{y}$ that is $1$-consistent and better than $f(n,\outliers)$-robust as in \eqref{eq:SC-LB-robustness}.
Consider the profiles depicted in Figure \ref{fig:SC-LB-pred-z>1-unbounded-n-big}.

First, assume that $n-\outliers$ is odd and consider the profile in Figure \ref{fig:SC-LB-pred-z>1-unbounded-n-big-a} with $\delta =0$, $d_1 > d_2+d_3$ and $0 < d_2 < d_3$. 
There are $\frac{n-\outliers-1}{2} = \outliers + \theta$, with $\theta \ge 0$, agents with a location of 0, $z$ agents with a location of $d_1$, $1 + \theta$ agents with a location of $d_1+d_2$, and $\outliers$ agents with a location of $d_1+d_2+d_3$. 
It is easy to see that the optimal location of the facility is at $d_1 + d_2$, as the optimal location is either at $d_1$ or at $d_1 + d_2$ by Lemma \ref{fact:optRange}. 
Consider a perfect prediction $\hat{y}$ located at $d_1+d_2$.
As $\mathcal{M}$ is 1-consistent, it must be that $\mathcal{M}$ places the facility at $y = \hat{y}$, as depicted in Figure \ref{fig:SC-LB-pred-z>1-unbounded-n-big-a}.

Now, consider the profile in Figure \ref{fig:SC-LB-pred-z>1-unbounded-n-big-a} with $\delta =1$ in which one agent with a location of 0 moved to $d_1$. 
As $\mathcal{M}$ is strategyproof, it must be that $\alg$ remains unchanged by Corollary \ref{lem:SP-property-M}.
The same reasoning holds when considering $\delta = 2, 3, \ldots, \outliers + \theta -1$ consecutively. 

Finally, consider the profile in Figure \ref{fig:SC-LB-pred-z>1-unbounded-n-big-b} in which all $z + \theta$ agents with a location of 0 moved to $d_1$. 
In this case, the optimal location $\opt$ is located at $d_1$ with a utilitarian social cost of $(1 + \theta) d_2$.  
Again as $\mathcal{M}$ is strategyproof, it must be that $\alg$ remains unchanged by Corollary \ref{lem:SP-property-M}.
The minimum social cost of $\alg$ is $(2\outliers + \theta) d_2$, as $d_3 > d_2$. 
And as $\opt \neq \hat{y}$, this leads to a robustness guarantee of $\frac{(2\outliers + \theta) d_2}{(1 + \theta) d_2} = \frac{2\outliers + \theta}{1 + \theta} = \frac{n+\outliers-1}{n-3\outliers+1}$, as $\theta = \frac{n-\outliers-1}{2} - \outliers$. This contradicts that $\mathcal{M}$ is better than $f(n,\outliers)$-robust as in \eqref{eq:SC-LB-robustness}.

If $n-\outliers$ is even\footnote{Note that for $z=1$, $f(n,z) =1$ which is a trivial bound on the robustness guarantee.}, consider the same profile in Figure \ref{fig:SC-LB-pred-z>1-unbounded-n-big-a} with $\delta =0$, but with $z - 1$ agents with a location of $d_1$ and with $\theta = \frac{n - \outliers}{2} - \outliers \ge 0$.
The optimal location is still located at $d_1+d_2$ and given a perfect prediction $\hat{y} = d_1+d_2$, $\mathcal{M}$ must still place the facility at $y = \hat{y}$ in order to be $1$-consistent. 
Again, moving all $z + \theta$ agents with a location of 0 consecutively to $d_1$, this leads to the profile depicted in Figure \ref{fig:SC-LB-pred-z>1-unbounded-n-big-b} with $2z + \theta -1$ agents with a location at $d_1$. And as $\mathcal{M}$ is strategyproof, $\mathcal{M}$ must again place the facility $y$ at $d_1 + d_2$ by Corollary \ref{lem:SP-property-M}. 
In this case, the optimal location $\opt$ is located at $d_1$ with a utilitarian social cost of $(1 + \theta) d_2$, and the minimum social cost of $\alg$ is $(2\outliers + \theta -1) d_2$, as $d_3 > d_2$.
As $\opt \neq \hat{y}$, this leads to a robustness guarantee of $\frac{(2\outliers + \theta -1) d_2}{(1 + \theta) d_2} = \frac{2\outliers + \theta -1}{1 + \theta} = \frac{n+\outliers-2}{n-3\outliers+2}$, resulting in a contradiction. This concludes the proof. 
\end{proof}

\subsubsection{Mechanism \InRange}

There is a subtlety why a strategyproof mechanism cannot achieve bounded robustness and $1$-consistency for $n < 3z$. 
Consider a profile $\vec{x}$ with 5 locations and $z = 2$ outliers.
By Lemma \ref{fact:optRange}, there exists an optimal location that is equal to $x_{\sigma(2)}$, $x_{\sigma(3)}$ or $x_{\sigma(4)}$. However, if there is only one optimal solution $\opt = x_{\sigma(4)}$, then $\sigma(2) \notin S^{\star}(\vec{x},z)$, making it potentially impossible to relate the optimal objective value to the objective value when placing the facility at $x_{\sigma(2)}$ (when $x_{\sigma(2)}$ is much smaller than $x_{\sigma(3)}$).
Whereas if the number of agents increases by 1 to $n=6=3z$, there always exists an optimal location that is equal to either $x_{\sigma(3)}$ or $x_{\sigma(4)}$. 
Whichever of these two locations is optimal, the other location will be regarded in the optimal set of non-outliers, i.e., $\sigma(3), \sigma(4) \in S^{\star}(\vec{x},z)$, and therefore the optimal objective value can always be related to the minimum objective value when placing the facility at any of these two locations.  

So given input $(\vec{x}, z, \hat{y})$, a mechanism $\mathcal{M}$ is ensured to have a bounded robustness guarantee if $\mathcal{M}$ never chooses a location of the facility that is smaller than $x_{\sigma(z+1)}$ or larger than $x_{\sigma(n-z)}$. 
Our mechanism \InRange, introduced below, satisfies this property by choosing the prediction $\hat{y}$ as the location of the facility if (i) $x_{\sigma(z+1)} \le \hat{y} \le x_{\sigma(n-z)}$, (ii) if $\hat{y}$ is larger than or equal to the smallest value in the set $O$ (Lemma \ref{fact:optRange}), and (iii) if $\hat{y}$ is smaller than or equal to the largest value in the set $O$.
Note that this interval is constructed in the first two lines of \InRange\ and that this interval is well defined. 
Additionally, note that for $n\ge 3z$, (i) is subsumed by (ii) and (iii).
\medskip
 
\begin{mechanism}[H]
\SetAlgoLined
\SetAlgoCaptionSeparator{}
\SetAlCapNameFnt{\scshape}
\DontPrintSemicolon
\caption{\InRange($\vec{x},z,\hat{y}$)}
let $\ell = \max \{ \cl{\frac{n - \outliers +1}{2}}, \outliers +1 \}$ and $r = \min \{ \cl{\frac{n - \outliers}{2}} + \outliers , n - \outliers \} $ \label{line:size-interval}\\
\lIf{$x_{\sigma( \ell)} \le \hat{y} \le x_{\sigma( r)}$\label{line:if-in-range}}{\Return $\hat{y}$}
\lIf{$\hat{y} < x_{\sigma( \ell)} $ }{\Return $x_{\sigma( \ell)}$ \label{line:left-range}}
\lElse{\Return $x_{\sigma(r)}$ \label{line:right-range}}
\end{mechanism}

\medskip
The main result of this section is the following theorem. Recall that the error measure $\eta(\vec{x}, z, \hat{y})$ is defined as the ratio of the minimum utilitarian social cost for the location $\hat{y}$ and the optimal utilitarian social cost.

\begin{restatable}{theorem}{fivefour} \label{th:SC-UB-n>=3z}
Let $n \ge 3 \outliers$. Then, for any input $(\vec{x}, z, \hat{y})$ with $\eta(\vec{x}, z, \hat{y}) \le \eta$, $\InRange$ is strategyproof and achieves an approximation guarantee for the utilitarian objective of:
\begin{equation} \label{th:eq:SC-UB-n>=3z}
    f(n,\outliers, \eta) = \begin{cases}
        \min \{ \eta, \frac{ n + z - 1 }{ n - 3z +1 } \}, & \text{if $n-\outliers$ is odd,} \\
         \min \{ \eta, \frac{ n + z - 2 }{ n - 3z +2 } \}, & \text{otherwise.}
    \end{cases} 
\end{equation}
In particular, \InRange\ is 1-consistent and $f(n,\outliers)$-robust as in \eqref{eq:SC-LB-robustness}, which is best possible. 
\end{restatable}

In order to prove the approximation guarantee of \InRange, we introduce some axillary notation. 
Consider the input $(\vec{x}, z, \hat{y})$ and a location $y$ of the facility.
Define $i(\opt)$ such that $x_{\sigma(i(\opt))} = \opt$ (note that such a $\opt$ always exists).
We define $i(y) = i(\opt)$ if $y \in [\opt_{\ell}, \opt_{r}]$.
If there exists a location in $\vec{x}$ at which the facility is located, i.e., $\exists k \in [n]$ such that $y = x_{\sigma(k)}$, we define $i(y)$ as the maximum (minimum) index such that $x_{\sigma(i(y))} = y$ if $\opt < y$ ($\opt > y$).
Otherwise, if $\opt < y$, we define $i(y)$ such that $x_{\sigma(i(y))} > y$ and $x_{\sigma(i(y)-1)} < y$.
Symmetrically, if $\opt > y$, we define $i(y)$ such that $x_{\sigma(i(y))} < y$ and $x_{\sigma(i(y)+1)} > y$.
Finally, we define $\delta$ as the difference in indices: 
\begin{equation} \label{eq:error2-def}
\delta =  |i(\opt) - i(y)|.
\end{equation}
Note that for \InRange, it holds that $\delta \in \{0, 1, \ldots, \outliers \}$ if $n-\outliers$ is odd and $\delta \in \{0, 1, \ldots, \outliers -1\}$ if $n-\outliers$ is even by construction and by Lemma \ref{fact:optRange}.
If $\delta =0$, the optimal social cost for $y$ and $y^{\star}$ coincide and as $\delta$ grows, the number of locations between $y$ and $y^{\star}$ potentially increases and potentially leads to a worse approximation guarantee. 

\begin{proof}[Proof of Theorem~\ref{th:SC-UB-n>=3z}]
Consider arbitrary input $(\vec{x}, z, \hat{y})$ and let $y = \InRange(\vec{x}, z, \hat{y})$.
We first show that \InRange\ is strategyproof. 
Consider a unilateral deviation of an agent $i \in N$ from $x_i =p_i$ to $x'_i$. 
First, consider the case that $\alg = \hat{y} \in [x_{\sigma( \ell)} , x_{\sigma( r)}]$. 
The deviation of agent $i$ can only potentially move $y$ to the right (left) if $x_i \le x_{\sigma( \ell)}$ ($x_i \ge x_{\sigma( \ell)}$), but this would increase the cost of agent $i$. 
Secondly, consider the case that $\alg = x_{\sigma( \ell)} > \hat{y}$. 
Also in this case the deviation of agent $i$ can only move $y$ to the right (resp. left) if $x_i \le x_{\sigma( \ell)}$ (resp. $x_i \ge x_{\sigma( \ell)}$), but this would increase the cost of agent $i$. 
Similar reasoning holds for $\alg = x_{\sigma(r)} < \hat{y}$.

We now show that \InRange\ is $f(n,\outliers, \eta)$-approximate as in \eqref{th:eq:SC-UB-n>=3z}.
For $n \ge 3 \outliers$ it holds that $\ell = \cl{\frac{n - \outliers +1}{2}}$ and $r = \cl{\frac{n - \outliers}{2}} + \outliers$ and so, $\ell \le r$.
Additionally, it therefore follows by Lemma \ref{fact:optRange} that \InRange\ is 1-consistent. Note that if $n-z$ is even and the perfect prediction $\hat{y} \notin [x_{\sigma(\ell)}, x_{\sigma(r)}]$, \InRange\ still returns an optimal solution by construction by returning $x_{\sigma(\ell)}$ or $x_{\sigma(r)}$.
Now consider an imperfect prediction $\hat{y}$ with $\hat{y} \in [x_{\sigma(\ell)}, x_{\sigma(r)}]$. In this case \InRange\ returns $\hat{y}$ and by definition of $\eta(\vec{x},\outliers,\hat{y})$ and as $\eta(\vec{x},\outliers,\hat{y}) \le \eta$, it follows that: 
\[
\OPTSC(\InRange(\vec{x}, z, \hat{y})) = \eta(\vec{x},\outliers,\hat{y}) \cdot \OPTSC^{\star}(\vec{x},\outliers) \le \eta \cdot \OPTSC^{\star}(\vec{x},\outliers) = f(n,z,\eta) \cdot \OPTSC^{\star}(\vec{x},\outliers).
\]
Here, the last equality follows from the reasoning below.

Consider any imperfect prediction $\hat{y}$.
We will show that $\OPTSC(y, \vec{x},z) \le f(n,z,\eta) \cdot \OPTSC^{\star}(\vec{x},\outliers)$. 
First, consider the case that $n-\outliers$ is odd and additionally, assume w.l.o.g. that $\opt < y$. 
We upper bound the social cost of $y$ in two ways. 
First, as in the proof of Theorem \ref{th:SCdeter}, we evaluate the social cost of $y$ w.r.t. the locations that $\opt$ accounts for, i.e., for locations $x_i$ with $i \in S^{\star}$.
Secondly, we upper bound the social cost of $y$ by considering the social cost of $y' = i(y)$ w.r.t. the locations $x_i$ with $i \in S^{\star}$. 
Note that this upper bounds the social cost as, in the case of $\opt < y$, it holds that $y \le y'$ and there are more locations $i \in S^{\star}$ with $x_i < y$ than there are with $x_i \ge y$.
Together, this leads to first inequality in \eqref{eq:SC:UB:z>1-proof} below. 

As $n - \outliers$ is odd, there are $\frac{n - \outliers -1}{2}$ locations $x_i$ with $i \in S^{\star}$ and $\sigma(i) < i(\opt)$ and $\frac{n - \outliers -1}{2}$ locations $x_i$ with $i \in S^{\star}$ and $\sigma(i) > i(\opt)$, as depicted in Figure \ref{fig:SC-UB-pred-z>1}. 
Furthermore, let $\delta =  |i(\opt) - i(y)| = |i(\opt) - i(y')|$ as in \eqref{eq:error2-def}.
Then, by Lemma \ref{fact:optRange}, there are $\delta - 1 \ge 0$ locations $x_i$ with $i \in S^{\star}$ and $i(\opt) < \sigma(i) < i(y')$, as $\hat{y}$ is not a perfect prediction, which is also depicted in Figure \ref{fig:SC-UB-pred-z>1}. 
Again as in the proof of Theorem \ref{th:SCdeter}, we upper bound the approximation guarantee by moving all the locations $x_i$ with $i \in S^{\star}$ and $\sigma(i) < \sigma(y')$ to $\opt$, and all the locations $x_i$ with $i \in S^{\star}$ and $\sigma(i) > \sigma(y')$ to $y'$, leading to:

\begin{equation} \label{eq:SC:UB:z>1-proof}
\frac{\OPTSC(\InRange(\vec{x}, z, \hat{y}))}{\OPTSC^{\star}(\vec{x}, \outliers)} 
\le \frac{\sum_{i \in S^{\star}} |\alg' - x_i|}{\sum_{i \in S^{\star}} |y^{\star} - x_i|} 
\le \frac{ \frac{n - \outliers -1}{2} + \delta }{\frac{n-\outliers-1}{2} + 1 -\delta}
\le \frac{ n + \outliers -1}{n - 3 \outliers +1}.
\end{equation}
Here, the last inequality follows as $\delta \le z$ if $n-z$ is odd. 
Additionally, note that if $\hat{y} \notin [x_{\sigma(\ell)}, x_{\sigma(r)}]$, as $\opt < y$ it must be that $\hat{y} > x_{\sigma(r)} =y$. 
In this case $\eta \ge \frac{\OPTSC(y, \vec{x}, z)}{\OPTSC^{\star}(\vec{x},z)}$, as the set of non-outliers minimizing the social cost are equal for $y$ and $\hat{y}$, and $y$ is the median of this set. 
Note that for $\opt > \alg$, similar reasoning holds (by considering Figure \ref{fig:SC-UB-pred-z>1} when multiplied by $-1$). Similar reasoning also holds when $n-\outliers$ even, i.e., when there are $\frac{n-\outliers}{2}$ locations left of $\opt$ and $\frac{n-\outliers}{2} -1  -\delta$ locations right of $\alg'$ in Figure \ref{fig:SC-UB-pred-z>1}, leading to:
\[
\frac{\OPTSC(\InRange(\vec{x}, z, \hat{y}))}{\OPTSC^{\star}(\vec{x}, \outliers)} 
\le \frac{ \frac{n - \outliers}{2} + \delta }{\frac{n-\outliers}{2} -\delta}
\le \frac{ n + \outliers-2 }{ n - 3 \outliers +2}, 
\]
where the last inequality follows as $\delta \le z-1$ if $n-z$ is even. 
Note that optimality of the consistency-robustness trade-off follows by Theorem \ref{lem:SC:LB-prediction-z>1-bounded}, concluding the proof.
\end{proof}

\begin{figure}[t]
\centering
\begin{subfigure}[b]{0.98\linewidth}
    \centering
    \begin{tikzpicture}
        \draw[thick] (-4.5,0) -- (10.5,0);
        \node at (-3.3,0.8) {$n-\outliers$ even:};
        \node at (-3.3,-0.7) {$n-\outliers$ odd:};
        
        \filldraw[black] (-1.5,0) circle (2pt) node[below, yshift=-0.1cm]{ };
        \filldraw[black] (-0.5,0) circle (2pt) node[below, yshift=-0.1cm]{ };
        \draw [thick, decorate,decoration={brace,amplitude=10pt,mirror},xshift=0.4pt,yshift=-0.4pt](-1.6,-0.1) -- (-0.4,-0.1) node[black,midway,yshift=-0.6cm] {$\frac{n-\outliers-1}{2}$};
        \draw [thick, decorate,decoration={brace,amplitude=10pt},xshift=0.4pt,yshift=-0.4pt](-1.6,0.1) -- (-0.4,0.1) node[black,midway,yshift=0.7cm] {$\frac{n-\outliers}{2}$};

        \filldraw[black] (1.5,0) circle (2pt) node[below, yshift=-0.1cm]{ };
        \draw[thick] (1.5,0) -- (1.5,0.5) node[above, xshift=0cm]{$\opt$};
        
        \filldraw[black] (2.66,0) circle (2pt) node[below, yshift=-0.1cm]{ };
        \filldraw[black] (3.33,0) circle (2pt) node[below, yshift=-0.1cm]{ };
        \draw [thick, decorate,decoration={brace,amplitude=10pt,mirror},xshift=0.4pt,yshift=-0.4pt](2.56,-0.1) -- (3.43,-0.1) node[black,midway,yshift=-0.6cm, xshift=-0.1cm] {$\delta -1$};
        \draw [thick, decorate,decoration={brace,amplitude=10pt},xshift=0.4pt,yshift=-0.4pt](2.56,0.1) -- (3.43,0.1) node[black,midway,yshift=0.7cm, xshift=-0.1cm] {$\delta -1$};
        
        \filldraw[black] (4.5,0) circle (2pt) node[below, yshift=-0.1cm]{ };
        \draw[thick] (4.5,0) -- (4.5,0.5) node[above, xshift=0cm]{$\alg'$};
        
        \filldraw[black] (6.5,0) circle (2pt) node[below, yshift=-0.1cm]{ };
        \filldraw[black] (7.5,0) circle (2pt) node[below, yshift=-0.1cm]{ };
        \draw [thick, decorate,decoration={brace,amplitude=10pt,mirror},xshift=0.4pt,yshift=-0.4pt](6.4,-0.1) -- (7.6,-0.1) node[black,midway,yshift=-0.6cm, xshift=0.1cm] {$\frac{n-\outliers-1}{2} - \delta $};
        \draw [thick, decorate,decoration={brace,amplitude=10pt},xshift=0.4pt,yshift=-0.4pt](6.4,0.1) -- (7.6,0.1) node[black,midway,yshift=0.7cm] {$\frac{n-\outliers}{2} - 1 -\delta$};

    \end{tikzpicture}
    \label{fig:SC-UB-pred-z>1-a}
\end{subfigure}
\caption{Part of the profile used in the proof of Theorem \ref{th:SC-UB-n>=3z}.}
\label{fig:SC-UB-pred-z>1}
\end{figure}
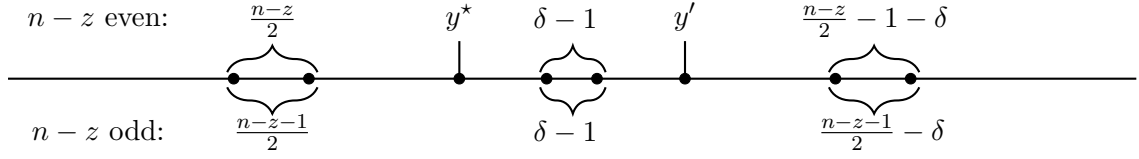

\subsubsection{Fine-Grained Analysis for $n < 3z$}

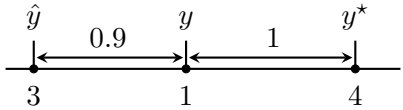
\begin{figure}[b]
    \centering
    \begin{tikzpicture}[scale=0.75]
        \draw[thick] (-0.5,0) -- (6.5,0);
        \filldraw[black] (0,0) circle (2pt) node[below, yshift=-0.1cm]{$3$};
        \draw[thick] (0,0) -- (0,0.5) node[above, xshift=0cm]{$\hat{y}$};
        \filldraw[black] (2.7,0) circle (2pt) node[below, yshift=-0.1cm]{$1$};
        \draw[thick] (2.7,0) -- (2.7,0.5) node[above, xshift=0cm]{$y$};
        \filldraw[black] (5.7,0) circle (2pt) node[below, yshift=-0.1cm]{$4$}; 
        \draw[thick] (5.7,0) -- (5.7,0.5) node[above, xshift=0cm]{$\opt$};
        \draw[stealth-stealth,thick](0.05,0.2) -- (2.65,0.2)  node[color=black, above, xshift=-1cm]{$0.9$};
        \draw[stealth-stealth,thick](2.75,0.2) -- (5.65,0.2)  node[color=black, above, xshift=-1.05cm]{$1$};
    \end{tikzpicture}
\caption{Profile with $n<3z$ for which $\InRange(\vec{x}, z, \hat{y})=y$ has a worse approximation guarantee than $\eta(\vec{x}, z, \hat{y})$.}
\label{fig:error-fails-n<3z}
\end{figure}

We consider the approximation guarantee of \InRange\ separately for $n \in \{2\outliers +1,2\outliers +2, \ldots, 3\outliers-1 \}$, as it cannot be stated as $f(n,z,\eta)$ of Theorem \ref{th:SC-UB-n>=3z}.
We elaborate on this with the following illustrative example. 
Consider the profile $\vec{x}$ with eight location, $z=3$ outliers and $\hat{y}$ as depicted in Figure \ref{fig:error-fails-n<3z}.
Note that the optimal social cost is 1, and that the minimum social cost of $\hat{y}$ is 2.8, leading to 
$\eta(\vec{x}, z, \hat{y}) =2.8$. 
However, the minimum social cost of the location $y$ chosen by \InRange, i.e., $y = \InRange(\vec{x}, z, \hat{y})$, is 3.7. 
And so, the approximation guarantee cannot be expressed as the minimum of the error and the robustness guarantee.  
This is because for $n<3z$, it is no longer true that the set of non-outliers minimizing the social cost is equal for $y$ and $\hat{y}$, and that $y$ is the median of this set. 
Recall that in this example, $\hat{y}$ is contained in the set $O$ of potential optimal locations defined by Lemma \ref{fact:optRange}, and returning the prediction if it is bounded by the smallest and largest values of $O$ would lead to an approximation guarantee of $\eta(\vec{x}, z, \hat{y})$, but this mechanism would fail to achieve any bounded robustness guarantee.

\paragraph{Fine-Grained Analysis for $n < 3z$.}
In order to derive a fine-grained approximation guarantee of \InRange\ for these values of $n$ and $z$, and not just a robustness guarantee, we will use $\delta$ as defined in \eqref{eq:error2-def} more prominently. 
Intuitively, a mechanisms potentially achieves a good approximation guarantee if it chooses a location close to an optimal location $y^{\star}$.
If the location $y$ chosen by the mechanism is not optimal, then $\delta -1$ indicates the number of agents with a location between $y$ and $y^{\star}$, and so $\delta$ can be interpreted as an error measure of the mechanism.
We therefore introduce some additional axillary notation to define this error measure for perfect and imperfect predictions.

For $n \in \{ 2 \outliers +1, 2\outliers +2, \ldots, 3 \outliers -1 \}$ it holds that $\ell = \max \{ \cl{\frac{n - \outliers +1}{2}}, \outliers +1 \} =  \outliers +1$ and $r = \min \{ \cl{\frac{n - \outliers}{2}} + \outliers , n - \outliers \} = n - \outliers$, so $\ell = \outliers +1 \le n - \outliers = r$.
Therefore, \InRange\ is only 1-consistent is some cases and in all other cases, the consistency guarantee depends on how much the thresholds $\ell$ and $r$ differ from the thresholds of Lemma \ref{fact:optRange}. 
We define this by determining the difference $\delta^{c}$ in the $k$-th order statistics used for the thresholds: 
\[
\delta^{c} = \outliers +1 - \bigg \lceil \frac{n - \outliers +1}{2} \bigg \rceil = \bigg \lceil \frac{n - \outliers}{2} \bigg \rceil  + \outliers  -( n - \outliers ).
\]
Note that the difference of the lower and upper thresholds is equal.
On the other hand, the robustness guarantee improves compared to the guarantee in \eqref{eq:SC-LB-robustness}.
This is due to the change in thresholds which leads to a smaller upper bound of $\delta$. 
Namely, in this case $\delta \le \delta^{r}$ with:
\[
\delta^{r} 
= \bigg \lceil \frac{n - \outliers}{2} \bigg \rceil  + \outliers - (\outliers + 1)
= n - \outliers- \bigg \lceil \frac{n - \outliers +1}{2} \bigg \rceil.
\]
Note that $\delta \le \delta^{c} + \delta^{r}$, and that the consistency guarantee is always better than the robustness guarantee as $\delta^{c} \le \delta^{r}$. We now present the approximation guarantee of \InRange\ for $n < 3z$. 

\begin{restatable}{theorem}{fivefive}\label{th:SC-UB-n<3z}
Let $n \in \{2 \outliers +1, 2 \outliers +2, \ldots, 3 \outliers -1 \}$.
Then, for any input $(\vec{x}, z, \hat{y})$, $\InRange$ is strategyproof and achieves an approximation guarantee for the utilitarian objective of:
\begin{equation} \label{th:eq:SC-UB-n<3z}
    f(n,\outliers, \delta) = \begin{cases}
        \frac{ \frac{n - \outliers -1}{2} + \delta }{\frac{n-\outliers-1}{2} +1 - \delta }, & \text{if $n-\outliers$ is odd,} \\[1em]
        \frac{ \frac{n - \outliers}{2} + \delta }{\frac{n-\outliers}{2} - \delta}, & \text{otherwise.}
    \end{cases} 
\end{equation}
In particular, \InRange\ is $f(n,\outliers,\delta^{c})$-consistent and $f(n,\outliers,\delta^{r})$-robust.
\end{restatable}

The proof of Theorem \ref{th:SC-UB-n<3z} follows analogously to the proof of Theorem \ref{th:SC-UB-n>=3z}.
Note that if the input $(\vec{x}, z, \hat{y})$ of \InRange\ is such that $\eta(\vec{x}, z, \hat{y}) \le \eta$, and $\hat{y} \in [x_{\sigma(\ell)},x_{\sigma(r)}]$ or $\exists y^{\star}_{\ell}, y^{\star}_{r}$ such that $\hat{y} \in [y^{\star}_{\ell}, y^{\star}_{r}]$ and $[y^{\star}_{\ell}, y^{\star}_{r}] \cap [x_{\sigma(\ell)},x_{\sigma(r)}]  \neq \emptyset$, \InRange\ achieves an approximation of $\eta$ for $n \in \{2 \outliers +1, 2 \outliers +2, \ldots, 3 \outliers -1 \}$. 

\paragraph{Confidence Parameter.} If $n \ge 3 \outliers$ then \InRange\ is 1-consistent but $f(n,\outliers)$-robust as in \eqref{eq:SC-LB-robustness}, and so the robustness guarantee is worse than the approximation guarantee of Theorem \ref{th:SCdeter} without predictions.
If there is some uncertainty about the quality of the prediction, one might want to settle for a worse consistency and a better robustness guarantee. 
For $n<3z$, we have already seen that \InRange\ can achieve a trade-off between consistency and robustness, by altering the thresholds that determine if the prediction is chosen as the location of the facility.
And so, we introduce a natural confidence parameter $\gamma$ that influences the thresholds defined in Line \ref{line:size-interval} of \InRange.
Namely, we add $\gamma$ to $\ell$ and subtract $\gamma$ from $r$. 
If $\outliers > 1$ and depending on the parity of $n$, $\gamma$ attains a feasible value if $\gamma \in \{ 0, 1, \ldots, \gamma_{\max} \}$, with $\gamma_{\max}$ defined as: 
\begin{equation*} \label{eq:gamma-range}
    \gamma_{\max} = \begin{cases}
        \frac{\outliers}{2} -1, & \text{if $n$ is even and $\outliers$ is even,} \\
        \fl{\frac{\outliers}{2}}, & \text{otherwise.}
    \end{cases} 
\end{equation*}
Here, $\gamma = 0$ can be interpreted as full confidence in the prediction, i.e., the thresholds defined in Line \ref{line:size-interval} of \InRange\ are not adjusted. 
As $\gamma$ increases, the left and right thresholds defining when the prediction $\hat{y}$ is chosen as the location of the facility become larger and smaller, respectively, modeling less confidence in the prediction. Less confidence, i.e., a larger $\gamma$, will lead to a worse consistency but an improved robustness. 
If $\gamma = \gamma_{\max}$, this translates to no confidence in the prediction and \InRange\ returns the median and achieves the approximation guarantee of Theorem \ref{th:SCdeter}.

\section{Conclusion} \label{sec:conclusion}

In this work, we initiated the study of mechanism design with outliers and considered facility location on the line, with and without predictions, as the test case for this new perspective.
As our results indicate, the problem becomes harder, in some cases significantly, when outliers are introduced. 
We would like to highlight that all of our deterministic strategyproof mechanisms are additionally {\em group}-strategyproof {\em in the strong sense}, meaning that there is no group of agents that can coordinate their declarations such that at least one of them decreases their cost while no one increases their cost. 
For the setting without predictions, this is a direct consequence of the characterization result of~\cite{Moulin80} because all our mechanisms are $k$-th order statistics. For the setting with predictions, this can be proven from first principles.

We conclude by highlighting a few research directions that originate from our model and results.
\begin{enumerate}
    \item Is there a general proof that shows a strong negative result for $z \ge \nicefrac{n}{2}$? 
    We believe this result can be extended to mechanisms that are strategyproof in expectation, and it would be interesting to see if this lower bound persists for a large family of mechanism design problems.
    \item Does there exist a different type of prediction that is provably useful for the egalitarian objective? 
    If {\em not}, is there an underlying property of the problem that forbids a positive result?
    \item Possibly, the most fruitful direction is to apply this new direction of outliers to other mechanism design problems. We envision that outliers can be incorporated in {\em any} mechanism design problem, at least without monetary transfers. Do other problems also exhibit the same phenomenon of the approximation guarantee degrading monotonically with the number of outliers? Does there exist a setting in which outliers actually help the mechanism designer?  We believe that answering any of these questions is equally exciting!
\end{enumerate}

\section*{Acknowledgements}

This work was supported by the Dutch Research Council (NWO) through its Open Technology Program, proj.~no.~18938 and the Gravitation Project NETWORKS, grant no.~024.002.003. It has also been funded by the European Union under the EU Horizon 2020 Research and Innovation Program under the Marie Skłodowska-Curie Grant Agreement, grant no.~101034253.
Furthermore, this work acknowledges the support of the EPSRC grant EP/X039862/1.

\bibliography{bib}

\end{document}